\documentclass[submission,copyright,creativecommons]{eptcs}
\providecommand{\event}{QPL 2020} 
\usepackage[utf8]{inputenc}
\usepackage{amsmath}
\usepackage{amssymb}
\usepackage{amsthm}
\usepackage{textcomp}
\usepackage{color} 
\usepackage{multicol}
\usepackage{graphicx}
\usepackage{tikz}
\usepackage{placeins}
\usepackage{stmaryrd}
\usepackage{framed}
\usepackage{tikzfig}
\usepackage{rotating}
\usepackage{changepage}
\usepackage[normalem]{ulem} 
\usepackage{mathtools} 
\DeclarePairedDelimiter\bra{\langle}{\rvert}
\DeclarePairedDelimiter\ket{\lvert}{\rangle}

\tikzstyle{env}=[copoint,regular polygon rotate=0,minimum width=0.2cm, fill=black]

\tikzstyle{every picture}=[baseline=-0.25em]
\tikzstyle{dotpic}=[scale=0.5]
\tikzstyle{diredges}=[every to/.style={diredge}]
\tikzstyle{dot graph}=[shorten <=-0.1mm,shorten >=-0.1mm,scale=0.6]
\tikzstyle{plot point}=[circle,fill=black,minimum width=2mm,inner sep=0]

\tikzstyle{braceedge}=[decorate,decoration={brace,amplitude=2mm,raise=-1mm}]
\tikzstyle{small braceedge}=[decorate,decoration={brace,amplitude=1mm,raise=-1mm}]
\tikzstyle{left hook arrow}=[left hook-latex]
\tikzstyle{right hook arrow}=[right hook-latex]

\tikzstyle{dtriangle}=[fill=yellow,draw=black,shape=isosceles triangle,shape border rotate=-90,isosceles triangle stretches=true,inner sep=1pt,minimum width=0.4cm,minimum height=3mm]
\tikzstyle{vtriang}=[fill=yellow,draw=black,shape=isosceles triangle,shape border rotate=180,isosceles triangle stretches=true,inner sep=1pt,minimum width=0.4cm,minimum height=3mm]
\tikzstyle{triangle}=[fill=yellow,draw=black,shape=isosceles triangle,shape border rotate=90,isosceles triangle stretches=true,inner sep=1pt,minimum width=0.4cm,minimum height=3mm]
\tikzstyle{H box}=[rectangle,fill=yellow,draw=black,xscale=1.0,yscale=1.0, inner sep=1.pt]
\tikzstyle{gbox}=[rectangle,fill=green,draw=black,xscale=1.0,yscale=1.0, inner sep=1.pt]
\tikzstyle{rbox}=[rectangle,fill=red,draw=black,xscale=1.0,yscale=1.0, inner sep=1.pt]

\tikzstyle{bn}=[circle,fill=black,draw=black,scale=.4]
\tikzstyle{wn}=[circle,fill=white,draw=black,scale=.6]
\tikzstyle{dn}=[circle,fill=none,draw=gray]

\tikzstyle{black dot}=[inner sep=0.7mm,minimum width=0pt,minimum height=0pt,fill=black,draw=black,shape=circle]

\tikzstyle{dot}=[black dot]
\tikzstyle{smalldot}=[inner sep=0.4mm,minimum width=0pt,minimum height=0pt,fill=black,draw=black,shape=circle]
\tikzstyle{white dot}=[dot,fill=white]
\tikzstyle{antipode}=[white dot,inner sep=0.3mm,font=\footnotesize]
\tikzstyle{smallwhitedot}=[smalldot,fill=white]
\tikzstyle{alt white dot}=[white dot,label={[xshift=3.07mm,yshift=-0.05mm,font=\footnotesize]left:$*$}]
\tikzstyle{gray dot}=[dot,fill=gray!40!white]
\tikzstyle{smallgraydot}=[smalldot,fill=gray!40!white]
\tikzstyle{box vertex}=[draw=black,rectangle]
\tikzstyle{small box}=[box vertex,fill=white]
\tikzstyle{whitebg}=[fill=white,inner sep=2pt]
\tikzstyle{graph state vertex}=[sg vertex,fill=black]

\tikzstyle{wide copoint}=[fill=white,draw=black,shape=isosceles triangle,shape border rotate=90,isosceles triangle stretches=true,inner sep=1pt,minimum width=1.5cm,minimum height=5mm]
\tikzstyle{wide point}=[fill=white,draw=black,shape=isosceles triangle,shape border rotate=-90,isosceles triangle stretches=true,inner sep=1pt,minimum width=1.5cm,minimum height=4mm]
\tikzstyle{very wide copoint}=[fill=white,draw=black,shape=isosceles triangle,shape border rotate=-90,isosceles triangle stretches=true,inner sep=1pt,minimum width=2.5cm,minimum height=4mm]
\tikzstyle{very wide empty copoint}=[draw=black,shape=isosceles triangle,shape border rotate=-90,isosceles triangle stretches=true,inner sep=1pt,minimum width=2.5cm,minimum height=4mm]
\tikzstyle{symm}=[ultra thick,shorten <=-1mm,shorten >=-1mm]

\tikzstyle{square box}=[rectangle,fill=white,draw=black,minimum height=5mm,minimum width=5mm,font=\small]
\tikzstyle{square gray box}=[rectangle,fill=gray!30,draw=black,minimum height=6mm,minimum width=6mm]
\tikzstyle{copoint}=[regular polygon,regular polygon sides=3,draw=black,scale=0.75,inner sep=-0.5pt,minimum width=7mm,fill=white]
\tikzstyle{point}=[regular polygon,regular polygon sides=3,draw=black,scale=0.75,inner sep=-0.5pt,minimum width=7mm,fill=white,regular polygon rotate=180]
\tikzstyle{gray point}=[point,fill=gray!40!white]
\tikzstyle{gray copoint}=[copoint,fill=gray!40!white]

\newcommand{\edgearrow}{{\arrow[black]{>}}}
\newcommand{\edgetick}{{\arrow[black,scale=0.7,very thick]{|}}}

\tikzstyle{diredge}=[->]
\tikzstyle{rdiredge}=[<-]
\tikzstyle{medium diredge}=[->]

\tikzstyle{short diredge}=[->]
\tikzstyle{halfedge}=[-)]
\tikzstyle{other halfedge}=[(-]
\tikzstyle{freeedge}=[(-)]
\tikzstyle{white edge}=[line width=5pt,white]
\tikzstyle{tick}=[postaction=decorate,decoration={markings, mark=at position 0.5 with \edgetick}]
\tikzstyle{small map edge}=[|-latex, gray!60!blue, shorten <=0.9mm, shorten >=0.5mm]
\tikzstyle{thick dashed edge}=[very thick,dashed,gray!40]
\tikzstyle{map edge}=[|-latex,very thick, gray!40, shorten <=1mm, shorten >=0.5mm]
\tikzstyle{tickedge}=[postaction=decorate,
  decoration={markings, mark=at position 0.5 with \edgetick}]
\tikzstyle{dirtickedge}=[postaction=decorate,
  decoration={markings, mark=at position 0.5 with \edgetick},
  decoration={markings, mark=at position 0.85 with \edgearrow}]
\tikzstyle{dirdoubletickedge}=[postaction=decorate,
  decoration={markings, mark=at position 0.4 with \edgetick},
  decoration={markings, mark=at position 0.6 with \edgetick},
  decoration={markings, mark=at position 0.85 with \edgearrow}]

\makeatletter
\newcommand{\boxshape}[3]{%
\pgfdeclareshape{#1}{
\inheritsavedanchors[from=rectangle] 
\inheritanchorborder[from=rectangle]
\inheritanchor[from=rectangle]{center}
\inheritanchor[from=rectangle]{north}
\inheritanchor[from=rectangle]{south}
\inheritanchor[from=rectangle]{west}
\inheritanchor[from=rectangle]{east}
\backgroundpath{
\southwest \pgf@xa=\pgf@x \pgf@ya=\pgf@y
\northeast \pgf@xb=\pgf@x \pgf@yb=\pgf@y

\@tempdima=#2
\@tempdimb=#3

\pgfpathmoveto{\pgfpoint{\pgf@xa - 5pt + \@tempdima}{\pgf@ya}}
\pgfpathlineto{\pgfpoint{\pgf@xa - 5pt - \@tempdima}{\pgf@yb}}
\pgfpathlineto{\pgfpoint{\pgf@xb + 5pt + \@tempdimb}{\pgf@yb}}
\pgfpathlineto{\pgfpoint{\pgf@xb + 5pt - \@tempdimb}{\pgf@ya}}
\pgfpathlineto{\pgfpoint{\pgf@xa - 5pt + \@tempdima}{\pgf@ya}}
\pgfpathclose
}
}}

\boxshape{NEbox}{0pt}{8pt}
\boxshape{SEbox}{0pt}{-8pt}
\boxshape{NWbox}{8pt}{0pt}
\boxshape{SWbox}{-8pt}{0pt}
\makeatother

\tikzstyle{map}=[draw,shape=NEbox,inner sep=7pt]
\tikzstyle{mapdag}=[draw,shape=SEbox,inner sep=7pt]
\tikzstyle{maptrans}=[draw,shape=SWbox,inner sep=7pt]
\tikzstyle{mapconj}=[draw,shape=NWbox,inner sep=7pt]

\tikzstyle{probs}=[shape=semicircle,fill=gray!40!white,draw=black,shape border rotate=180,minimum width=1.2cm]

\tikzstyle{arrs}=[-latex,font=\small,auto]
\tikzstyle{arrow plain}=[arrs]
\tikzstyle{arrow dashed}=[dashed,arrs]
\tikzstyle{arrow bold}=[very thick,arrs]
\tikzstyle{arrow hide}=[draw=white!0,-]
\tikzstyle{arrow reverse}=[latex-]
\tikzstyle{cdnode}=[]

\tikzstyle{gn}=[dot,fill=green,minimum width=0.3cm,inner sep=0pt]
\tikzstyle{rn}=[dot,fill=red,inner sep=0pt,minimum width=0.3cm]

\tikzstyle{rc}=[dot,thick,fill=white,draw = red,minimum width=0.3cm,inner sep=0pt]
\tikzstyle{gc}=[dot,thick,fill=white,draw= green,inner sep=0pt,minimum width=0.3cm]
\tikzstyle{bc}=[dot,thick,fill=white,draw= blue,minimum width=0.3cm]

\tikzstyle{label}=[circle,fill=white,minimum width=0.3cm]

\tikzstyle{clocklabel}=[dot,fill=yellow,draw=black,font=\tiny,inner sep=0.75pt]

\tikzstyle{rsn}=[circle split,draw,fill=red,font=\tiny,inner sep=0.75pt]
\tikzstyle{gsn}=[circle split,draw,fill=green,font=\tiny,inner sep=0.75pt]
\tikzstyle{bsn}=[circle split,draw,fill=blue,font=\tiny,inner sep=0.75pt]

\tikzstyle{rsc}=[circle split,thick,draw= red,draw,fill=white,font=\tiny,inner sep=0.75pt]
\tikzstyle{gsc}=[circle split,thick,draw= green,draw,fill=white,font=\tiny,inner sep=0.75pt]
\tikzstyle{bsc}=[circle split,thick,draw= blue,draw,fill=white,font=\tiny,inner sep=0.75pt]

\tikzstyle{cnot}=[fill=white,shape=circle,inner sep=-1.4pt]
\tikzstyle{wire label}=[font=\tiny, auto]

\tikzstyle{cdiag}=[matrix of math nodes, row sep=3em, column sep=3em, text height=1.5ex, text depth=0.25ex,inner sep=0.5em]
\tikzstyle{arrow above}=[transform canvas={yshift=0.5ex}]
\tikzstyle{arrow below}=[transform canvas={yshift=-0.5ex}]

\newtheorem{Th}{Theorem}[section]
\newtheorem{theorem}[Th]{Theorem}
 
\newtheorem{lemma}[Th]{Lemma}
\newtheorem{corollary}[Th]{Corollary}

\newtheorem{remark}[Th]{Remark}

\makeatletter
\newcommand{\vast}{\bBigg@{6.5}}
\makeatother

\newcommand{\vertrule}[1][1ex]{\rule{.4pt}{#1}}
\def\bR{\begin{color}{red}}  
\def\bB{\begin{color}{blue}}
\def\bM{\begin{color}{magenta}} 
\def\bC{\begin{color}{cyan}} 
\def\bW{\begin{color}{white}}
\def\bBl{\begin{color}{black}}
\def\bG{\begin{color}{green}}
\def\bY{\begin{color}{yellow}}
\def\e{\end{color}}
\newcommand{\bit}{\begin{itemize}}
\newcommand{\eit}{\end{itemize}\par\noindent}
\newcommand{\ben}{\begin{enumerate}} 
\newcommand{\een}{\end{enumerate}\par\noindent}
\newcommand{\beq}{\begin{equation}}
\newcommand{\eeq}{\end{equation}\par\noindent}
\newcommand{\beqa}{\begin{eqnarray*}}
\newcommand{\eeqa}{\end{eqnarray*}\par\noindent}
\newcommand{\beqn}{\begin{eqnarray}}
\newcommand{\eeqn}{\end{eqnarray}\par\noindent}

\title{ AND-gates in ZX-calculus: \\ Spider Nest Identities and QBC-completeness}
\author{Anthony Munson
\institute{Mathematical Institute \\University of Oxford}
\email{anthony.munson@queens.ox.ac.uk}
\and
Bob Coecke and  Quanlong Wang
\institute{Cambridge Quantum Computing Ltd.}
\email{bob.coecke, harny.wang@cambridgequantum.com}
}
\def\titlerunning{AND-gates in ZX-calculus: Spider Nest Identities and QBC-completeness}
\def\authorrunning{A. Munson, B. Coecke  \& Q. Wang}
\begin{document}
\maketitle

\begin{abstract}
In this paper we exploit the utility of the triangle symbol which has a complicated expression in terms of spider diagrams in ZX-calculus, and its role within the ZX-representation of AND-gates in particular.  First, we derive spider nest identities which are of key importance to recent developments in quantum circuit optimisation and T-count reduction in particular.  Then, using the same rule set, we prove a completeness theorem for quantum Boolean circuits (QBCs) whose rewriting rules can be directly used for a new method of T-count reduction. We give an algorithm based on this method and show that the results of our algorithm outperform the results of all the previous best non-probabilistic algorithms. 
\end{abstract}

\section{Introduction}
The ZX-calculus \cite{CD1,CD2} is a universal graphical language for qubit  theory, which comes equipped with simple rewriting rules that enable one to transform diagrams representing one quantum process into  another quantum process.  More broadly, it is the work-horse of categorical quantum mechanics,  which aims for a high-level formulation of quantum theory \cite{AC1,CKbook}.  

Recently ZX-calculus has been completed by Ng and Wang \cite{ng2017universal}, that is, provided with sufficient additional rules so that any equation between matrices in Hilbert space can be derived in ZX-calculus.  This followed  earlier completions  by Backens   for  stabiliser theory \cite{Backens} and one-qubit Clifford+T circuits \cite{Backens2},  and   by Jeandel, Perdrix and Vilmart for general Clifford+T theory \cite{jeandel2017complete}.  

This paper concerns with the `utility' of ZX-rules.  While in principle the rules in 
\cite{ng2017universal}  are  sufficient to any other rule, it is by no means guaranteed that it is in any way intuitive, easy, or even realistic to do so.  While the original rules of \cite{CD1,CD2} have been inherited by all ZX-calculi,   comparing \cite{ng2017universal} with  two more recent universal completeness results \cite{jeandel2018diagrammatic,Vilmart2018} one immediately notices that the additional rules that give completeness are entirely different in each of the papers, and their relationship is by no means obvious.   

Consequently, the new game in town is to match rules on their utility, and in some cases we will need to derive new rules.
In fact, the new rule of the completeness theorem of \cite{Vilmart2018} was in fact discovered/derived by two of the current authors with the particular purpose of quantum circuit simplification in mind \cite{DBLP:conf/rc/CoeckeW18}, and this utility preceded the completeness result.  One ingredient of the axiomatisation of  the original universal completeness result \cite{ng2017universal} is a new primitive of ZX-calculus, the triangle, which is one of the key components to make an algebraic axiomatisation  of ZX-calculus \cite{wangalg2020}.  Although it seems that  the triangle is an extension of the original form of ZX-calculus, it actually can be expressed in terms of original generators such as green and red spiders \cite{jeandel2017complete, ng2017universal}.  

We will explore one utility of the triangle, namely the role of standard rules governing the AND-gate when translated as ZX-rules.  In its simplest form the ZX-encoding of the AND-gate directly involves the triangle, this fact was firstly found in \cite{qwangkfclifford2018},  while its semantically equivalent form was already shown in \cite{CKbook}  (see Exercise 12.10 there). We in particular derive three results based on a same rule set:

{\bf Derivation of spider nest identities.} 
Recently ZX-calculus has been used to outperform all other methods in the area of circuit optimisation \cite{Aleks1, Aleks2, nielbianwang, nielbianwangtqc}.   Key to these for the purpose of T-gate reductions are so-called `phase gadgets' named by Kissinger and van de Wetering  \cite{Aleks1}, and were independently introduced by de Beaudrap and Wang, summarised in \cite{nielbianwang}.   Decomposing larger $\pi/4$-phase gadgets into smaller ones (less than 4 lines) is vital for reducing T-count further than reduction effect of phase gadget fusion. In this paper, we derive a general decomposition theorem for arbitrary phase gadget in terms of AND-gates. As a consequence, we obtain all the spider nest identities which play a key role in \cite{nielbianwang, nielbianwangtqc}. We note that the spider nest identities were derived in this way before the papers \cite{nielbianwang, nielbianwangtqc} were formed, while this paper was in the process of preparing. 

{\bf ZX-completeness for quantum Boolean circuits.}  Using the same rule set, we prove a completeness theorem for quantum Boolean circuits.   Circuit relations by Iwama, Kambayashi and Yamashita that achieve this have been known for a while \cite{iwama2002transformation}, also Cockett and Comfort have proved Iwama et al.'s rules in the symmetric monoidal category, TOF, generated by
the Toffoli gate and computational ancillary bits \cite{CockettComfort}.
We obtain our ZX-completeness result by proving Iwama et al.'s rules in the ZX-calculus with the rule set established in this paper, without any resort to having a Toffoli gate as the generator. This work does place results in circuit re-writing all under the umbrella of ZX-calculus. One particular advantage of this is the availability of automation tools which already has proved to be very useful for the above mentioned circuit optimisation results \cite{Aleks3}. 

{\bf Algorithm for reducing T-count of quantum circuits.}
Based on the rules of completeness for quantum Boolean circuits, we give an algorithm for reducing the T-count (number of T-gates) of a family of benchmark circuits called Galois field multipliers \cite{galoiscircuits}. We show in a table that the results of our algorithm outperform the results of all the previous best non-probabilistic algorithms. The technique in our algorithm can be applied to more benchmark circuits beyond Galois field multipliers.

{\bf Other work.} An early motivation for the GHZ/W-calculus \cite{CK}, now completed by Hadzihasanovic \cite{Amar, hadzihasanovic2017algebra} and known as ZW-calculus, was extended control operations, which is also one of the motivational upshots of the triangle symbol.  The triangle symbol entered the ZX-picture when ZW-calculus was translated to the ZX-context. After the ZX-encoding of the AND-gate with triangles was given in \cite{qwangkfclifford2018}, there came a very recent ZX-alike calculus called ZH-calculus \cite{EPTCS287.2}, which also allows easy representation of the AND-gate composed of fewer nodes than that of the ZX-encoding, however, it is not just about the denotation of the AND-gate, all the rules about the triangle  in this paper are also important for the utility of the AND-gate. 
  For now, ZX-calculus is (still) the umbrella under which quantum circuit optimisation is outperforming all competition, and were most completeness theorems have been stated, so a study of the AND-gate within this context is more than justified.

\section{ZX-calculus generators}


The ZX-calculus lives in a compact closed category whose objects are the natural numbers $\mathbb{N}$ and whose monoidal product is addition, $a\otimes b=a+b$. A general morphism $k\to l$ in this category is simply a $k$-input, $l$-output diagram generated, via finite sequential composition with outputs from one diagram connecting to inputs form another diagram  and tensor composition with diagrams in parallel, by the following elementary diagrams:
	\begin{center}
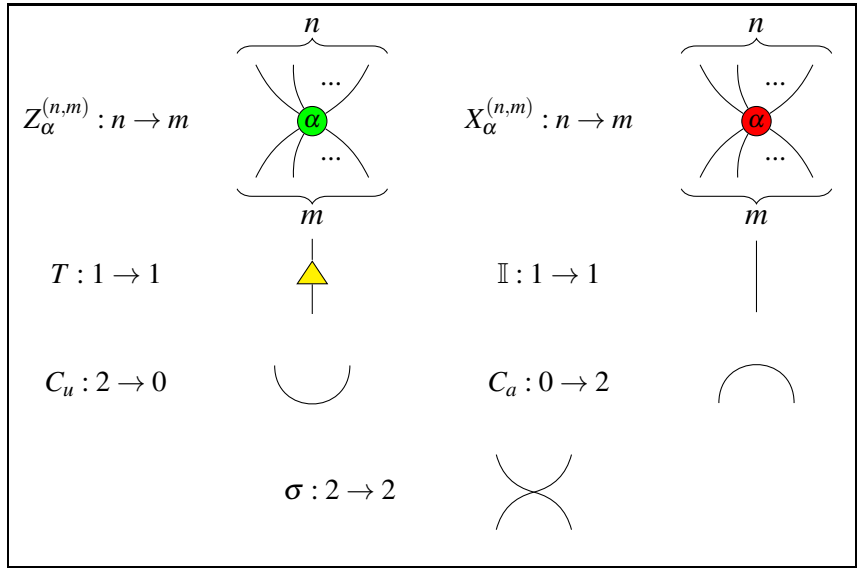
 	\begin{tabular}{| c c c c |}
  		\hline						  		$Z^{(n,m)}_\alpha:n\to m$	&\hspace{0.2cm}\input{TikZit/generatorGreenSpider.tikz}	&\hspace{0.5cm}$X^{(n,m)}_\alpha:n\to m$	&\hspace{0.2cm}\input{TikZit/generatorRedSpider.tikz}	\\
  		$T:1\to1$			&\hspace{0.2cm}\begin{tikzpicture}
	\begin{pgfonlayer}{nodelayer}
		\node [style=none] (0) at (0, 0.5) {};
		\node [style=none] (1) at (0, -0.5) {};
		\node [style=triangle] (2) at (0, -0) {};
	\end{pgfonlayer}
	\begin{pgfonlayer}{edgelayer}
		\draw (0.center) to (2);
		\draw (2) to (1.center);
	\end{pgfonlayer}
\end{tikzpicture}	&\hspace{0.5cm}$\mathbb{I}:1\to1$	&\hspace{0.2cm}\begin{tikzpicture}
	\begin{pgfonlayer}{nodelayer}
		\node [style=none] (0) at (0, 0.5) {};
		\node [style=none] (1) at (0, -0.5) {};
	\end{pgfonlayer}
	\begin{pgfonlayer}{edgelayer}
		\draw (0) to (1);
	\end{pgfonlayer}
\end{tikzpicture} 				\\
  												&	&	& \\
  		$C_u:2\to0$			&\hspace{0.2cm}\begin{tikzpicture}
	\begin{pgfonlayer}{nodelayer}
		\node [style=none] (0) at (-0.5, 0.25) {};
		\node [style=none] (1) at (0.5, 0.25) {};
	\end{pgfonlayer}
	\begin{pgfonlayer}{edgelayer}
		\draw [bend right=90, looseness=1.75] (0.center) to (1.center);
	\end{pgfonlayer}
\end{tikzpicture}			&\hspace{0.5cm}$C_a:0\to2$ 			&\hspace{0.2cm}\begin{tikzpicture}
	\begin{pgfonlayer}{nodelayer}
		\node [style=none] (0) at (-0.5, -0.25) {};
		\node [style=none] (1) at (0.5, -0.25) {};
	\end{pgfonlayer}
	\begin{pgfonlayer}{edgelayer}
		\draw [bend left=90, looseness=1.75] (0.center) to (1.center);
	\end{pgfonlayer}
\end{tikzpicture}					\\
  												&	&	& \\
	   \multicolumn{2}{|r}{  $\sigma:2\to2$ }&%
	\beginpgfgraphicnamed{TikZit/generatorSwap}
	\InputIfFileExists{TikZit/generatorSwap.tikz}{}{\input{./figures/TikZit/generatorSwap.tikz}}%
	\endpgfgraphicnamed
& \\								
  												&	&	& \\ \hline
	\end{tabular}
\end{center}
where 
$\alpha \in [0, 2\pi)$. Throughout this paper, all the diagrams are read from top to bottom, and non-zero scalars are ignored.  Note that the   triangle node generator is firstly introduced in \cite{jeandel2017complete} which  has  forms purely composed of green and red notes, here we show one of them as given in \cite{CKbook}: 
\begin{equation}\label{triangledecomposeeq}
	\beginpgfgraphicnamed{TikZit/triangledecompose}
	\InputIfFileExists{TikZit/triangledecompose.tikz}{}{\input{./figures/TikZit/triangledecompose.tikz}}%
	\endpgfgraphicnamed
 
\end{equation}
 This means the triangle is not totally external to the original ZX-calculus, but we won't use (\ref{triangledecomposeeq}) as a rule in this paper.
 
  Each generating diagram $D:k\to l$ above has a standard interpretation as a linear map between Hilbert spaces, $\llbracket D\rrbracket:(\mathbb{C}^2)^{\otimes k}\to(\mathbb{C}^2)^{\otimes l}$. By endowing each tensor factor $\mathbb{C}^2$ with its standard inner product and by identifying its elements $(1,0)^T$ and $(0,1)^T$ as the qubit Z-basis states $|0\rangle$ and $|1\rangle$, we can present each map $\llbracket D\rrbracket$ as a matrix:
\par\noindent
$\vast\llbracket$\hspace{0.1cm}\input{TikZit/generatorGreenSpider.tikz}$\vast\rrbracket$ = $|0\rangle^{\otimes m}\langle0|^{\otimes n}+e^{i\alpha}|1\rangle^{\otimes m}\langle1|^{\otimes n}$; \hspace{0.5cm}
$\vast\llbracket$\hspace{0.1cm}\input{TikZit/generatorRedSpider.tikz}$\vast\rrbracket$ = $|+\rangle^{\otimes m}\langle+|^{\otimes n}+e^{i\alpha}|-\rangle^{\otimes m}\langle-|^{\otimes n}$, \newline
where $|+\rangle=\frac{1}{\sqrt{2}}\big(|0\rangle+|1\rangle\big)$ and $|-\rangle=\frac{1}{\sqrt{2}}\big(|0\rangle-|1\rangle\big)$ are the qubit X-basis states:\par\noindent
$\Bigg\llbracket$\hspace{0.1cm}\begin{tikzpicture}
	\begin{pgfonlayer}{nodelayer}
		\node [style=none] (0) at (0, 0.5) {};
		\node [style=none] (1) at (0, -0.5) {};
		\node [style=triangle] (2) at (0, -0) {};
	\end{pgfonlayer}
	\begin{pgfonlayer}{edgelayer}
		\draw (0.center) to (2);
		\draw (2) to (1.center);
	\end{pgfonlayer}
\end{tikzpicture}$\Bigg\rrbracket$ = $\begin{pmatrix} 1&1\\ 0&1 \end{pmatrix}$; 
$\Bigg\llbracket$\hspace{0.1cm}\begin{tikzpicture}
	\begin{pgfonlayer}{nodelayer}
		\node [style=none] (0) at (0, 0.5) {};
		\node [style=none] (1) at (0, -0.5) {};
	\end{pgfonlayer}
	\begin{pgfonlayer}{edgelayer}
		\draw (0) to (1);
	\end{pgfonlayer}
\end{tikzpicture}$\Bigg\rrbracket$ = $\begin{pmatrix} 1&0\\ 0&1 \end{pmatrix}$; 
$\Bigg\llbracket$\hspace{0.1cm}\begin{tikzpicture}
	\begin{pgfonlayer}{nodelayer}
		\node [style=none] (0) at (-0.5, 0.25) {};
		\node [style=none] (1) at (0.5, 0.25) {};
	\end{pgfonlayer}
	\begin{pgfonlayer}{edgelayer}
		\draw [bend right=90, looseness=1.75] (0.center) to (1.center);
	\end{pgfonlayer}
\end{tikzpicture}$\Bigg\rrbracket$ = $\begin{pmatrix} 1&0&0&1 \end{pmatrix}$;

$\Bigg\llbracket$\hspace{0.1cm}\begin{tikzpicture}
	\begin{pgfonlayer}{nodelayer}
		\node [style=none] (0) at (-0.5, -0.25) {};
		\node [style=none] (1) at (0.5, -0.25) {};
	\end{pgfonlayer}
	\begin{pgfonlayer}{edgelayer}
		\draw [bend left=90, looseness=1.75] (0.center) to (1.center);
	\end{pgfonlayer}
\end{tikzpicture}$\Bigg\rrbracket$ = $\begin{pmatrix} 1\\ 0\\ 0\\ 1 \end{pmatrix}$;
$\Bigg\llbracket$\hspace{0.1cm}\input{TikZit/generatorSwap.tikz}$\Bigg\rrbracket$ = $\begin{pmatrix} 1&0&0&0\\ 0&0&1&0\\ 0&1&0&0\\ 0&0&0&1 \end{pmatrix}$; 
$\Bigg\llbracket$\hspace{0.1cm}\input{TikZit/generatorEmpty.tikz}$\Bigg\rrbracket$ = $1$, \newline
where 
	\beginpgfgraphicnamed{TikZit/generatorEmpty}
	\InputIfFileExists{TikZit/generatorEmpty.tikz}{}{\input{./figures/TikZit/generatorEmpty.tikz}}%
	\endpgfgraphicnamed
 represents an empty diagram.
We can extend these matrix interpretations to general ZX-diagrams simply by demanding that the composition operations for ZX-diagrams are compatible with those for matrices, 
namely that $\llbracket D_1\circ D_2\rrbracket=\llbracket D_1\rrbracket\circ\llbracket D_2\rrbracket$ and $ \llbracket D_1\otimes D_2\rrbracket=\llbracket D_1\rrbracket\otimes\llbracket D_2\rrbracket$
for every suitable pair of ZX-diagrams $D_1$ and $D_2$.

\section{ZX-calculus rules}


Here is the fragment of stabilizer-style ZX-calculus rules that we will make use of in this paper, where $\alpha, \beta \in [0,~2\pi), \kappa \in \{0, \pi\}$. All of these rules also hold when the colours red and green swapped. We also assume a meta rule that ``only topology/connectedness matters" \cite{CD2}.
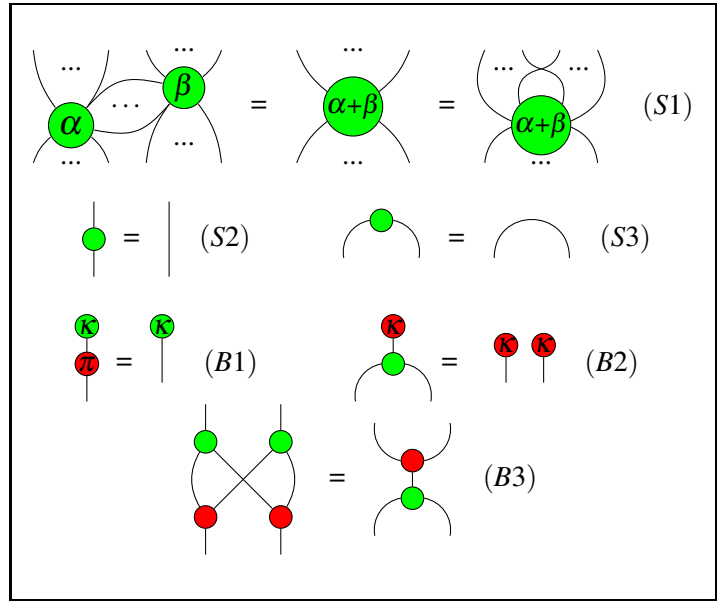
\begin{center}
\fbox{$\begin{array}{c}
\\ 
	\beginpgfgraphicnamed{TikZit/ruleS1}
	\InputIfFileExists{TikZit/ruleS1.tikz}{}{\input{./figures/TikZit/ruleS1.tikz}}%
	\endpgfgraphicnamed
\quad(S1)\\
\\ 
\begin{array}{ccc}
	\beginpgfgraphicnamed{TikZit/ruleS2}
	\InputIfFileExists{TikZit/ruleS2.tikz}{}{\input{./figures/TikZit/ruleS2.tikz}}%
	\endpgfgraphicnamed
\quad(S2) &  \quad& %
	\beginpgfgraphicnamed{TikZit/ruleS3}
	\InputIfFileExists{TikZit/ruleS3.tikz}{}{\input{./figures/TikZit/ruleS3.tikz}}%
	\endpgfgraphicnamed
\quad(S3)\\
\\ 
	\beginpgfgraphicnamed{TikZit/ruleB1snd}
	\InputIfFileExists{TikZit/ruleB1snd.tikz}{}{\input{./figures/TikZit/ruleB1snd.tikz}}%
	\endpgfgraphicnamed
\quad(B1) &  \quad &%
	\beginpgfgraphicnamed{TikZit/ruleB2snd}
	\InputIfFileExists{TikZit/ruleB2snd.tikz}{}{\input{./figures/TikZit/ruleB2snd.tikz}}%
	\endpgfgraphicnamed
\quad(B2)\\
\end{array}
\\ 
	\beginpgfgraphicnamed{TikZit/ruleB3}
	\InputIfFileExists{TikZit/ruleB3.tikz}{}{\input{./figures/TikZit/ruleB3.tikz}}%
	\endpgfgraphicnamed
\quad(B3)\\
\\ 
\end{array}$}
\end{center}


Noting that we haven't mentioned adjoints, we represent the transpose as follows:
 \begin{equation}\label{downtriangledef}
 \input{TikZit/downtriangle.tikz}
\end{equation}	
We will make use of the following simple triangle rules of the universal complete rules of \cite{ng2017universal}:
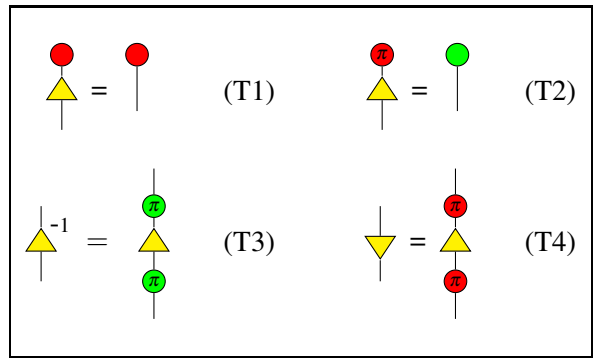
  \begin{center}
  	\begin{tabular}{| c c c c |}
	 \hline 
	 &	&	& \\
  		\input{TikZit/ruleT1.tikz}	&\hspace{0.2cm}(T1)	&\qquad
		\input{TikZit/ruleT2.tikz}	&\hspace{0.2cm}(T2) 	\\
  		&	&	& \\
  		\input{TikZit/definitionTriangleInverse2.tikz}	&\hspace{0.2cm}(T3)	&\qquad
		\input{TikZit/ruleT4.tikz}	&\hspace{0.2cm}(T4) 	\\
  		&	&	& \\ \hline
	\end{tabular}
     \end{center}
These rules represent the fact that the triangle breaks down of unitarity, or equivalently, the non-preservation of unitarity, more specifically, given that the red spider states $X^{(0,1)}_\alpha$ are the qubit Z-basis states:
\beq\label{Eq:Zbasis}
\left\llbracket\hspace{0.1cm}\begin{tikzpicture}
	\begin{pgfonlayer}{nodelayer}
		\node [style=rn, scale=1] (0) at (0, 0.25) {};
		\node [style=none] (1) at (0, -0.25) {};
	\end{pgfonlayer}
	\begin{pgfonlayer}{edgelayer}
		\draw (0) to (1.center);
	\end{pgfonlayer}
\end{tikzpicture}\hspace{0.1cm}\right\rrbracket = |0\rangle
	\qquad 
\left\llbracket\hspace{0.1cm}\begin{tikzpicture}
	\begin{pgfonlayer}{nodelayer}
		\node [style=rn, scale=1] (0) at (0, 0.25) {\scalebox{0.5}{$\pi$}};
		\node [style=none] (1) at (0, -0.25) {};
	\end{pgfonlayer}
	\begin{pgfonlayer}{edgelayer}
		\draw (0) to (1.center);
	\end{pgfonlayer}
\end{tikzpicture}\right\rrbracket = |1\rangle
\eeq
we have:
\begin{itemize}
\item Rules (T1) and (T2) show that the triangle turns orthonormal states into unbiased states. 
\item  Rules (T3) and (T4) show that the inverse of the triangle is not equal to its adjoint. 
\end{itemize}
This breakdown of unitarity greatly extends expressiveness, for example, for control operations where a basis can be used to switch between unbiased vectors.

\section{ZX rules on AND-gates}
In this section, we demonstrate  how to obtain ZX rules on AND-gates. To save space, we have placed all the Propositions in the appendix.

The interpretation (\ref{Eq:Zbasis}) also enables us to view each red spider state as a diagrammatic representation of a classical bit value. As such, the standard ZX-rules (S1) and (B2) easily demonstrate how the classical processes COPY and XOR are represented as ZX-diagrams:
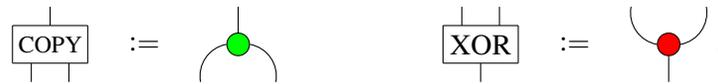
\begin{center}
	\input{TikZit/definitionCOPY.tikz}
	\hspace{0.5cm} \qquad \hspace{0.5cm}
	\input{TikZit/definitionXOR.tikz}.
\end{center}
Also, the unit of the copy process can be seen as a classical deleting process:
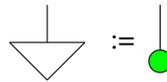
\[
	\beginpgfgraphicnamed{TikZit/delete}
	\InputIfFileExists{TikZit/delete.tikz}{}{\input{./figures/TikZit/delete.tikz}}%
	\endpgfgraphicnamed

\] 
Proposition 1 shows how the green spider $X^{(1,m)}_0$ defines the $m$-output generalized COPY  for any $m\geq 1$:
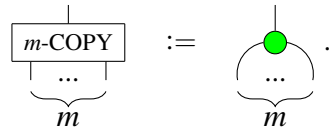
\begin{center}
	\input{TikZit/definitionGeneralCOPY.tikz}.
\end{center}
The quantum AND gate is a quantum version of the basic digital logic AND gate,  which has the following form 
 $$\ket{0}\bra{00}+\ket{0}\bra{01}+\ket{0}\bra{10}+\ket{1}\bra{11}.$$
The first ZX-form of quantum AND gate was given in \cite{CKbook} with slash boxes which represent the diagram on the right side of (\ref{triangledecomposeeq}). It is first realised in \cite{qwangkfclifford2018} that  the 
 quantum AND gate can be simply represented by triangle nodes, thus has a neat representation for its generalisation, the $n$-AND \cite{QWthesis}:
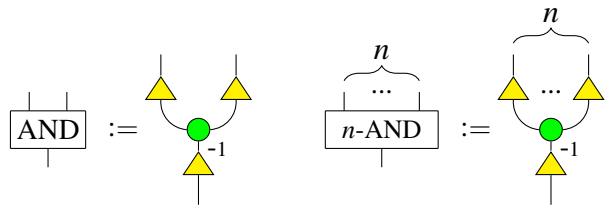
\begin{center}
	\input{TikZit/definitionAND.tikz}
\qquad
	\input{TikZit/definitionGeneralAND.tikz}.
\end{center}
It is straightforward to see that these diagrams give a proper representation of the $n$-input generalized AND. For, by Proposition 6, the 0-AND process simply returns $|1\rangle$; and, for all $n\geq0$, the $(n+1)$-AND process returns $|1\rangle$ whenever its input bit string consists solely of $|1\rangle$ states:
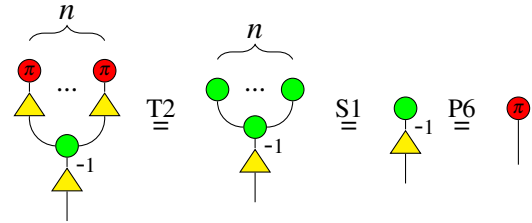
\begin{center}
	\input{TikZit/definitionGeneralAND_Analysis1.tikz}
\end{center}
and returns $|0\rangle$ whenever its input bit string contains at least one $|0\rangle$ state: 
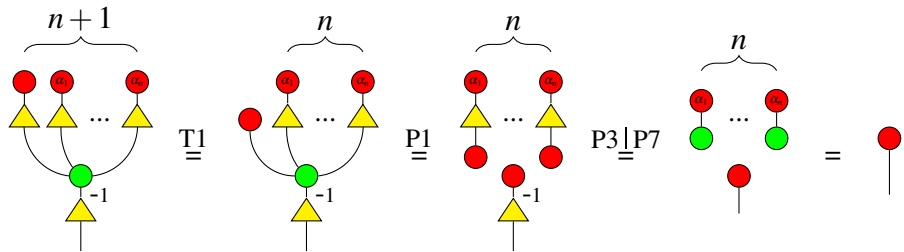
\begin{center}
	\input{TikZit/definitionGeneralAND_Analysis2.tikz}
\end{center}
where $\alpha_i \in \{0, \pi\}, i=1, \cdots, n$.
By identifying the COPY, XOR, and AND  in ZX-diagrams, we have demonstrated that the ZX-calculus has all the machinery necessary to represent Boolean algebra. In light of this fact, we expect that the Boolean identities $(p\oplus q)\cdot r=(p\cdot r)\oplus(q\cdot r)$ and $p\cdot p=p$ to hold:
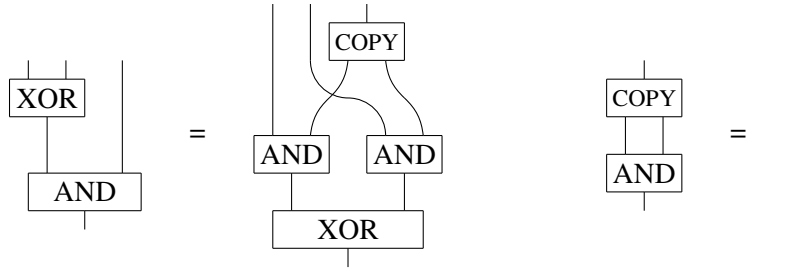
\begin{center}
	\input{TikZit/ruleA1_Analysis.tikz}
	\hspace{0.5cm} \qquad \hspace{0.5cm}
	\input{TikZit/ruleA2_Analysis.tikz}
\end{center}
Furthermore, AND  is actually a function map, so a homomorphism for  COPY and its unit (see e.g., \cite{CKbook}): 
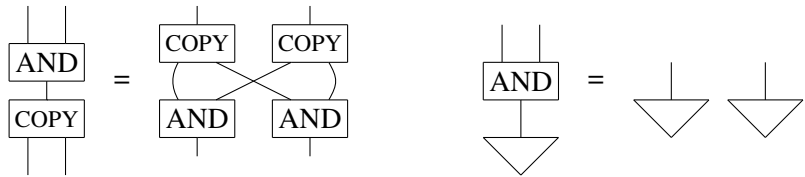
\[
	\beginpgfgraphicnamed{TikZit/andbialgebra}
	\InputIfFileExists{TikZit/andbialgebra.tikz}{}{\input{./figures/TikZit/andbialgebra.tikz}}%
	\endpgfgraphicnamed
 \hspace{0.5cm} \qquad \hspace{0.5cm}  %
	\beginpgfgraphicnamed{TikZit/andunit}
	\InputIfFileExists{TikZit/andunit.tikz}{}{\input{./figures/TikZit/andunit.tikz}}%
	\endpgfgraphicnamed
  
\]
In terms of ZX-diagrams, these equations for AND translate as:
  
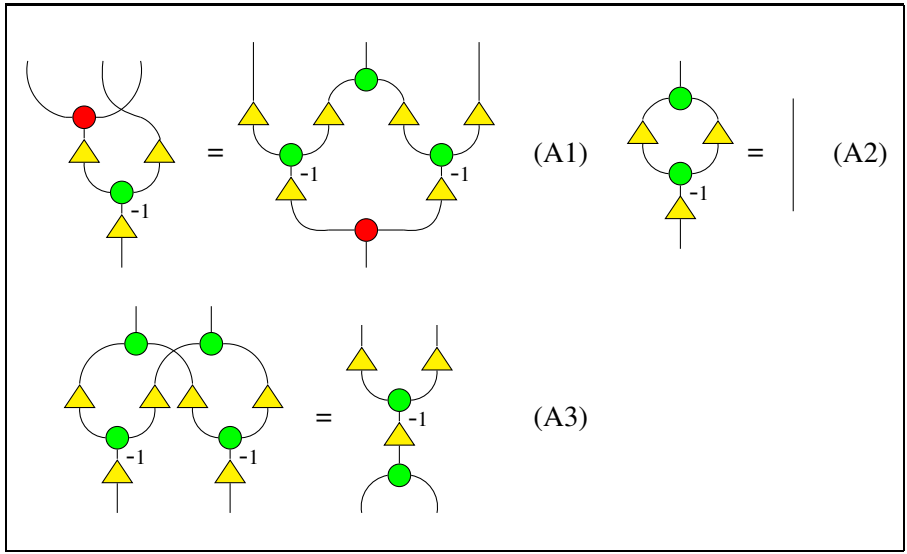
\begin{center}
  	\begin{tabular}{| c c c c |}
	 \hline 
	 &	&	& \\
\input{TikZit/ruleA1.tikz}	& (A1) \quad	&   \input{TikZit/ruleA22.tikz} & (A2)
			\\
  		&	&	& \\
\input{TikZit/ruleA3.tikz}	& (A3) \quad	&                   &   \\  
  		&	&	& \\ \hline
	\end{tabular}  
     \end{center}     
The ZX form of the second equality for the AND process as a function map is not listed as a rule, since it can be derived from other rules (see Proposition 8 in the Appendix). 
The rules (A1) and (A3) 
 have the following useful generalisations whose proofs can be found in the Appendix: \\
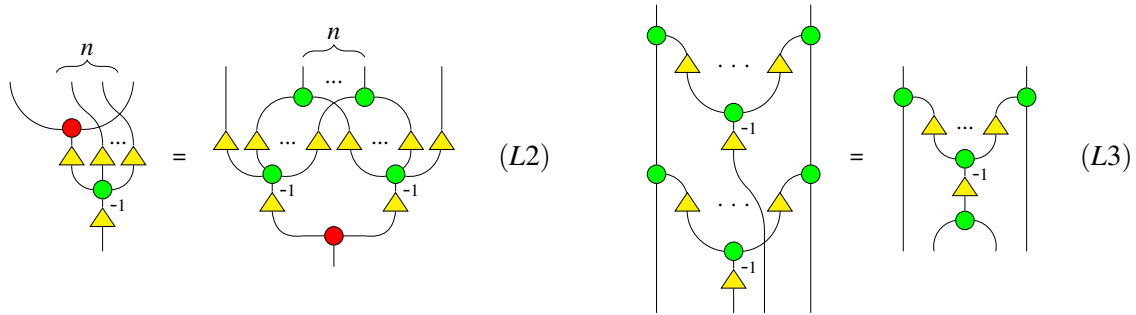
\[
 \scalebox{0.82}{\input{TikZit/appendixL2.tikz}} \hspace{0.5cm} (L2) \qquad \hspace{0.5cm}
\scalebox{0.82}{\input{TikZit/appendixL3.tikz}} \hspace{0.5cm} (L3)
\]

We note that  (A1) and (A3)  can be derived from another set of rules which give an algebraic axiomatisation of ZX-calculus \cite{wangalg2020}.

\section{AND-gates and spider nest identities} 

In this section, using the  AND-gate representation and rules presented in the previous section, we derive so-called spider nest identities which have been shown to be of vital importance  in reducing T-count with ZX-calculus \cite{nielbianwang, nielbianwangtqc}.  For this purpose, in addition to the rules listed in the figures of previous section, we assume the following equality with
 AND-gates: 
\begin{equation}\label{gadget-and}
 \scalebox{0.92}{\input{TikZit/gadget-andrule.tikz}}
\end{equation}
where $\alpha \in [0, 2\pi)$. It can be checked that the two diagrams on both sides of (\ref{gadget-and}) have the same standard interpretation up to some scalar. 
If we plug  %
	\beginpgfgraphicnamed{TikZit/ket0}
	\begin{tikzpicture}
	\begin{pgfonlayer}{nodelayer}
		\node [style=rn] (0) at (0, 0) {};
		\node [style=none] (1) at (0, -0.25) {};
	\end{pgfonlayer}
	\begin{pgfonlayer}{edgelayer}
		\draw (0) to (1.center);
	\end{pgfonlayer}
\end{tikzpicture}}%
	\endpgfgraphicnamed
  and %
	\beginpgfgraphicnamed{TikZit/bra0}
	\begin{tikzpicture}
	\begin{pgfonlayer}{nodelayer}
		\node [style=none] (0) at (0, 0.25) {};
		\node [style=rn] (1) at (0, 0) {};
	\end{pgfonlayer}
	\begin{pgfonlayer}{edgelayer}
		\draw (1) to (0.center);
	\end{pgfonlayer}
\end{tikzpicture}}%
	\endpgfgraphicnamed
 respectively onto the bottom and top of the leftmost lines of both sides of (\ref{gadget-and}), then we get:
\begin{equation}\label{gadget-and-2line}
\input{TikZit/gadget-andrule2line.tikz}
\end{equation}
which represents the controlled phase shift gate $\Lambda Z_{\alpha}$ \cite{CD2} in two different ways.
 
The following theorem generalises (\ref{gadget-and-2line})  from $2$ qubits to $n$ qubits in terms of AND-gates, as is proved in the Appendix.  
\begin{theorem}\label{decomposetheorem}
For any  $\alpha \in [0, 2\pi),~ n\geq 2$, we have:
\begin{equation}\label{phasedecomfull}
	\beginpgfgraphicnamed{TikZit/phasegadgetdecomfull}
	\InputIfFileExists{TikZit/phasegadgetdecomfull.tikz}{}{\input{./figures/TikZit/phasegadgetdecomfull.tikz}}%
	\endpgfgraphicnamed
  
\end{equation}
where for each $k\in\{2, \cdots, n\} $ lines of the RHS of (\ref{phasedecomfull}), there locates one and only one k-AND gate plugged with a phase with angle $\alpha_k= (-1)^k2^{k-2}\alpha$. Clearly $\alpha_{k+1}=-2\alpha_k$.  
\end{theorem}
Using the  the terminology of  \cite{Aleks1}, we denote the following diagram as a phase $n$-gadget:
$$%
	\beginpgfgraphicnamed{TikZit/phasegadgetdefinition}
	\InputIfFileExists{TikZit/phasegadgetdefinition.tikz}{}{\input{./figures/TikZit/phasegadgetdefinition.tikz}}%
	\endpgfgraphicnamed
  $$
where $n \geq 1$.
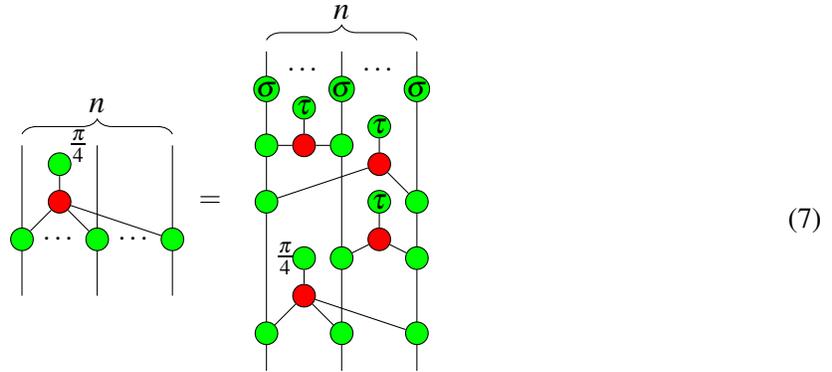
\begin{corollary}
\begin{equation}\label{gadgetdec}
\input{TikZit/tphasegadgetdecomrect.tikz}
\end{equation}
where $n \geq 3$, for the RHS of (\ref{gadgetdec}), on each line, there is a $1$-gadget with phase angle $\sigma=\frac{(n-2)(n-3)\pi}{8}$, for every two of the $n$ lines, there is a $2$-gadget with phase angle $\tau=\frac{(3-n)\pi}{4}$, and for every three of the $n$ lines, there is a $3$-gadget with phase angle $\frac{\pi}{4}$. 
\end{corollary}
\begin{proof}
Since the angle of the phase gadget is $\frac{\pi}{4}$, we let $\alpha=-\frac{\pi}{2}$ in  (\ref{phasedecomfull}). Then on the RHS of  (\ref{phasedecomfull}), $\alpha_k=0, \forall k \geq 4$. Thus all the k-AND gates plugged with phase
 $\alpha_k$ and $k \geq 4$ become identities, which means we just have phase-plugged 2-AND gates and 3-AND gates left in the RHS of  (\ref{phasedecomfull}). Now if we replace all phase-plugged 2-AND gates and 3-AND gates by the phase gadgets in  (\ref{gadget-and-2line})  and (\ref{gadget-and}) respectively, after phase gadget fusion, we then have the equality (\ref{gadgetdec}).  
\end{proof}
We note that the above proof can be generalised to more general cases if we consider the phase angle  on the left of (\ref{gadgetdec}) to be $\frac{\pi}{2^k}$ for any $k\geq 0$. 

Now if we plug the inverse of the phase gadget on the LHS of (\ref{gadgetdec}) to its both sides, then we get the spider nest identities \cite{nielbianwang, nielbianwangtqc} as follows:
$$%
	\beginpgfgraphicnamed{TikZit/spidernestid}
	\InputIfFileExists{TikZit/spidernestid.tikz}{}{\input{./figures/TikZit/spidernestid.tikz}}%
	\endpgfgraphicnamed
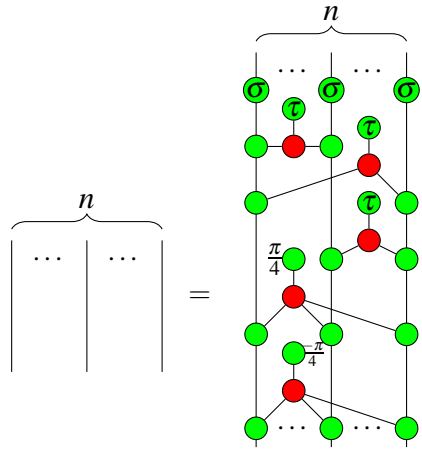
  $$

\section{CNOT-based Quantum Boolean Circuits}
In \cite{iwama2002transformation}, a CNOT gate is defined as an n-bit Toffoli gate \cite{toffoli} denoted by $[t, C]$, where $\ket{x_t}$ is a target bit and  $\ket{x_k}$ is a control bit with $k\in C \subseteq \{1, \cdots, n\}.$
 A quantum Boolean circuit is defined as follows: \medskip
\bit
\item A \textit{quantum Boolean circuit} of size $M$ over qubits $|x_1\rangle,...,|x_N\rangle$ is a sequence of CNOT gates\\ $[t_1,C_1]\cdots[t_i,C_i]\cdots[t_M,C_M]$ where $1\leq t_i\leq N$ and $C_i\subseteq\{1,...,N\}$. \eit
A quantum Boolean circuit is said to be  \textit{proper} and
to compute a Boolean function $f(x_1, \cdots, x_n)$ iff (i) the input state of the circuit $S_0=\ket{a_1}\ket{a_2}\cdots \ket{a_{n+1}}\ket{0}\cdots \ket{0}$, (ii) the final state of the circuit $S_M=\ket{a_1}\ket{a_2}\cdots \ket{a_{n+1}\oplus f(x_1, \cdots, x_n)}\ket{0}\cdots \ket{0}$. There are six transformation rules introduced in \cite{iwama2002transformation} for quantum Boolean circuits:\newline  
(1) $[t_1, C_1]\cdot[t_1, C_1]\iff\varepsilon$, \newline  
(2) $[t_1, C_1]\cdot[t_2, C_2]\iff[t_2, C_2]\cdot[t_1, C_1]$,  if $t_1\not\in C_2$ and $t_2\not\in C_1$,\newline
(3) $[t_1, C_1]\cdot[t_2, C_2]\iff[t_2, C_2]\cdot[t_1, C_1]\cdot[t_1, C_1\cup C_2-\{t_2\}]$ if $t_1\not\in C_2$ and $t_2\in C_1$, \newline
(4) $[t_1, C_1]\cdot[t_2, C_2]\iff[t_2, C_1\cup C_2-\{t_1\}]\cdot[t_2, C_2]\cdot[t_1, C_1]$ if $t_1\in C_2$ and $t_2\not\in C_1$, \newline
(5) $[t_1, \{c_1\}]\cdot[t_2, C_2\cup\{c_1\}]\iff[t_1, \{c_1\}]\cdot[t_2, C_2\cup\{t_1\}]$ if $t_1>n+1$ and no CNOT$_{t_1}$ before $[t_1, \{c_1\}]$, \newline
(6) $[t, C]\iff\varepsilon$ if there is an integer $i$ such that $i\in C$, $i>n+1$, and there is no CNOT$_i$ before $[t,C]$,\newline
where $\varepsilon$ represent the `identity' gate and $\iff$ denote a transformation. 

As a main result of \cite{iwama2002transformation}, this set of transformation rules is proved to be complete for proper quantum Boolean circuits: any two proper quantum Boolean circuits which are equivalent (computing the same Boolean function) can be derived from each other via these transformation rules. This result is claimed by the authors to be ``quite non-trivial and theoretically interesting on its own."

In this section, we represent each quantum Boolean circuit by ZX diagrams and then obtain the ZX-completeness result for quantum Boolean circuits by proving in ZX-calculus the complete set of six transformation rules presented in   \cite{iwama2002transformation}.  Moreover, we give an algorithm for T-count reduction based on these transformation rules and show its outperformance in contrast with results of state-of-the-art algorithms. 






 Now we give diagrammatic representation for CNOT gates. The ZX-diagram for a CNOT gate is simply (see e.g., \cite{QWthesis}):
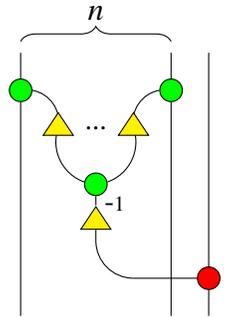
\begin{center}
	\input{TikZit/definitionGeneralCNOT.tikz}.
\end{center}
As a check to this representation, the reader can verify that in the $n=0$ and $n=1$ cases, this representation reduces to the NOT gate and the standard CNOT gate, as expected: 
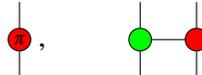
\begin{center}
	\begin{tikzpicture}
	\begin{pgfonlayer}{nodelayer}
		\node [style=rn, scale=1] (0) at (0, 0) {\scalebox{0.65}{$\pi$}};
		\node [style=none] (1) at (0, 0.5) {};
		\node [style=none] (2) at (0, -0.5) {};
	\end{pgfonlayer}
	\begin{pgfonlayer}{edgelayer}
		\draw (1.center) to (0);
		\draw (0) to (2.center);
	\end{pgfonlayer}
\end{tikzpicture},\hspace{1cm} \input{TikZit/definitionGeneralCNOT_Analysis2.tikz}
\end{center}
We have now developed the necessary machinery to assert and prove the completeness of our ZX-rules for quantum Boolean circuits. Let $ZX\vdash$ denote what can be inferred from the ZX rules listed in this paper: (S1), (S2), (S3), (B1), (B2), (B3), (T1), (T2), (T3), (T4), (A1), (A2), (A3).

\begin{theorem}\label{qbccomplete} 
If $D_1$ and $D_2$ are any two CNOT-based circuits expressed as ZX-diagrams such that  $\llbracket D_1 \rrbracket=\llbracket D_2\rrbracket$, then $ZX\vdash D_1=D_2$. 
\end{theorem}
The proof is given in the appendix.
\begin{remark}
Theorem \ref{qbccomplete} does not directly follows from the universal completeness of ZX-calculus, because the rules used for the proof of universal completeness are beyond the rules we presented in this paper.
\end{remark}

Below we give an algorithm for reducing T-count (number of T-gates, denoted as $\#$T) of quantum circuits based on the transformation rules. The transformation rules (3) and (4) are all about how to commute two $CNOT_i$ gates when exactly one target meets the other's control bit. Especially, when one of the two gates is the standard CNOT gate and the other an n-bit Toffoli gate, then there will be a new  n-bit Toffoli gate produced after their commuting. Based on this fact and the algorithms in \cite{nielbianwang} and \cite{nielbianwangtqc}, we have the following algorithm for reducing the T-count of a family of benchmark circuits---Galois field multipliers which are composed of  standard CNOT gates  in the middle and standard Toffoli gates on both left and right ends  \cite{galoiscircuits}:
\begin {enumerate}
\item
 Input a circuit of Galois field multiplier.

\item Turn the circuit into a ZX-diagram.

\item Start from the left most standard CNOT gate, commute it through to the left or right end of the circuit (outside of all standard Toffoli gates) with standard Toffoli gates using the  transformation rule (3) or (4). The commuting direction (to left or right) depends on the number of Toffoli gates it needs to commute with, the direction with fewer commutations will be chosen; if both directions have the same commuting number, then we choose the left direction.

\item Repeat the above step until all the CNOT gates are moved out.

\item Turn all the Toffoli gates into phase gadgets  and do the phase gadget fusion, using the method from \cite{nielbianwang, nielbianwangtqc}.

\item Output the circuit form and give the T-count.
\end {enumerate}
The following table presents a comparison of the results of our algorithm with the previous best (non-probabilistic) algorithms for reducing T-count.

\begin{table}[h!]
  \begin{center}
    \caption{T-count performance comparison.}
    \label{tab:table1}
    \begin{tabular}{c|c|c|c|c|c} 
\textbf{Circuit} & \textbf{Initial } $\#T$ & $\#T$ \textbf{in} \cite{NRSCM-2018} &$\#T$ \textbf{in} \cite{Aleks1}& $\#T$ \textbf{in} \cite{ZhangChen-2019}& \textbf{ Ours}\\
      \hline
    GF($2^4$)-mult&  112 & 68 & 68&68&62\\
     \hline
      GF($2^5$)-mult&  175 & 115 & 115&115&97\\
     \hline
     GF($2^6$)-mult&  252 & 150 & 150&150&131\\
     \hline
      GF($2^7$)-mult&  343 & 217 & 217&217&183\\
     \hline
     GF($2^8$)-mult&  448 & 264 & 264&264&263\\
     \hline
      GF($2^9$)-mult&  567 & 351 & -&351&299\\
     \hline
       GF($2^{10}$)-mult&  700 & 410 & -&410&361\\
     \hline
     GF($2^{16}$)-mult&  1792 & 1040 & -&1040&1038\\
     \hline
    \end{tabular}
  \end{center}
\end{table}
 \FloatBarrier
The circuits we obtained by running our algorithms can be found at \cite{Xiaoning}, they have been verified by \textit{feynver}  \cite{Amy-2018}.

As one can see from Table \ref{tab:table1}, our results outperform all the results of previous non-probabilistic algorithms. Note that the spider nest identities can be separately applied with  after running the above described algorithm, yet 
 we didn't apply these identities in our algorithm just because we want to see its own performance. 
 We would  expect better results if we use those identities. Probabilistic method in addition to our algorithm is in the similar situation as the spider nest identities. We also point out that the technique used in our algorithm can be applied to more benchmark circuits beyond the family of circuits of Galois field multipliers.
\section*{Acknowledgements} 

The authors thank Xiaoning Bian for his kind help on running numerical experiments on the algorithm proposed in this paper.  AM thanks the Mellon Mays Undergraduate Fellowship Program for its generous support of his research.  BC and QW are grateful for an EPSRC IAA in collaboration with Cambridge Quantum Computing Ltd. QW is supported by AFOSR grant FA2386-18-1-4028.









%
\bibliographystyle{eptcs}
\bibliography{ZXANDgate}    

\newpage
\section*{Appendix: Propositions, Lemmas and Proofs}
In this appendix we give  the Propositions, Lemmas and Proofs. Note that  in the proofs we just indicate the main rules which are used in derivations, for the sake of simplicity.
\subsubsection*{Proposition 1}

	For $m\geq 1$, \hspace{0.5cm} \input{TikZit/appendixP1.tikz} \hspace{0.5cm} (P1) \bigskip

	Proof. Using S1 we can decompose an $R_Z^{(1,m)}$ diagram into an $(m-1)$-fold composition of $R_Z^{(1,2)}$ diagrams. Repeatedly applying B2 then yields the desired result. \bigskip
	
\subsubsection*{Proposition 2}

	\input{TikZit/appendixP2.tikz}	\hspace{0.5cm}(P2)

	Proof. \hspace{0.25cm} \input{TikZit/appendixP2_Proof.tikz}

\subsubsection*{Proposition 3}

	\input{TikZit/appendixP3.tikz}	\hspace{0.5cm}(P3)

	Proof. \hspace{0.25cm} \input{TikZit/appendixP3_Proof.tikz} \bigskip

\subsubsection*{Proposition 4}

	\input{TikZit/appendixP4.tikz}	\hspace{0.5cm}(P4)

	Proof. \hspace{0.25cm} \input{TikZit/appendixP4_Proof.tikz}

\subsubsection*{Proposition 5}

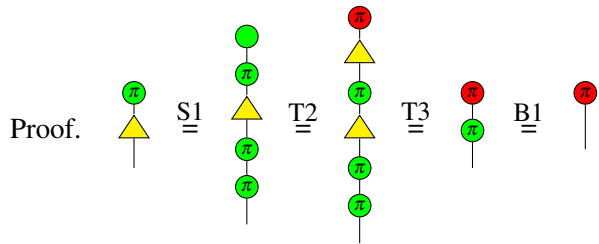
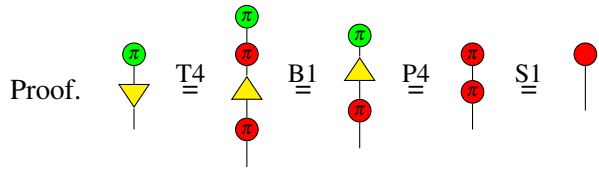
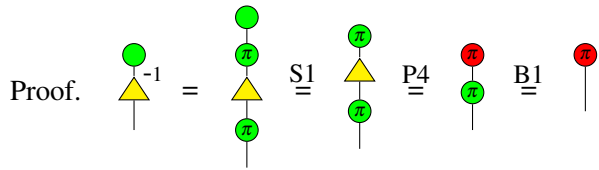
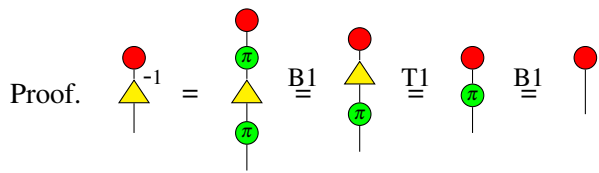
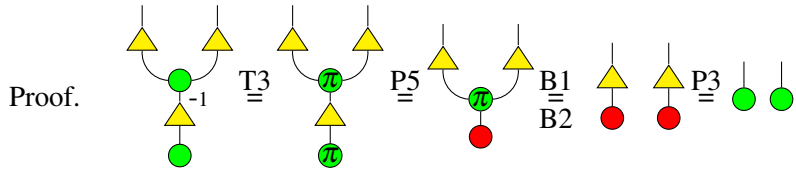
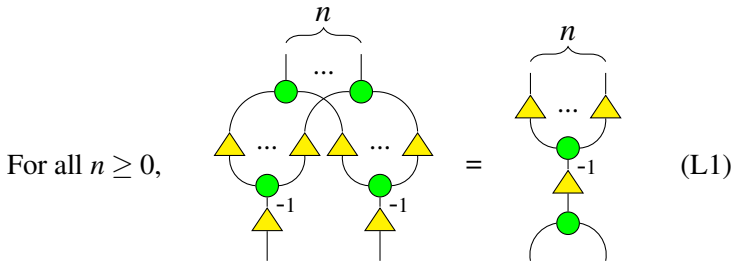
	\input{TikZit/appendixP5.tikz} \hspace{0.5cm}(P5)

	Proof. \hspace{0.25cm} \input{TikZit/appendixP5_Proof.tikz} \bigskip

\subsubsection*{Proposition 6}

	\input{TikZit/appendixP6.tikz} \hspace{0.5cm}(P6)

	Proof.  \hspace{0.25cm} \input{TikZit/appendixP6_Proof.tikz} \bigskip

\subsubsection*{Proposition 7}

	\input{TikZit/appendixP7.tikz} \hspace{0.5cm}(P7)

	Proof. \hspace{0.25cm} \input{TikZit/appendixP7_Proof.tikz} \bigskip

\subsubsection*{Proposition 8}
	\beginpgfgraphicnamed{TikZit/zxandunit}
	\InputIfFileExists{TikZit/zxandunit.tikz}{}{\input{./figures/TikZit/zxandunit.tikz}}%
	\endpgfgraphicnamed
   \hspace{0.5cm}(P8)

	Proof. \hspace{0.25cm} %
	\beginpgfgraphicnamed{TikZit/zxandunitprf}
	\InputIfFileExists{TikZit/zxandunitprf.tikz}{}{\input{./figures/TikZit/zxandunitprf.tikz}}%
	\endpgfgraphicnamed
   \bigskip

\subsubsection*{Lemma 1}

	For all $n\geq0$, \hspace{0.5cm} \input{TikZit/appendixL1.tikz} \hspace{0.5cm} (L1) \bigskip

	Proof. The $n=0$ case follows from B2 and P6, the $n=1$ case follows from (S2) and (T3), and the $n=2$ case is simply A1. The lemma is true in all other cases, since whenever it holds for some $n\geq2$, it is also valid for $n+1$: 
	$$ %
	\beginpgfgraphicnamed{TikZit/appendixL1_Proof}
	\InputIfFileExists{TikZit/appendixL1_Proof.tikz}{}{\input{./figures/TikZit/appendixL1_Proof.tikz}}%
	\endpgfgraphicnamed
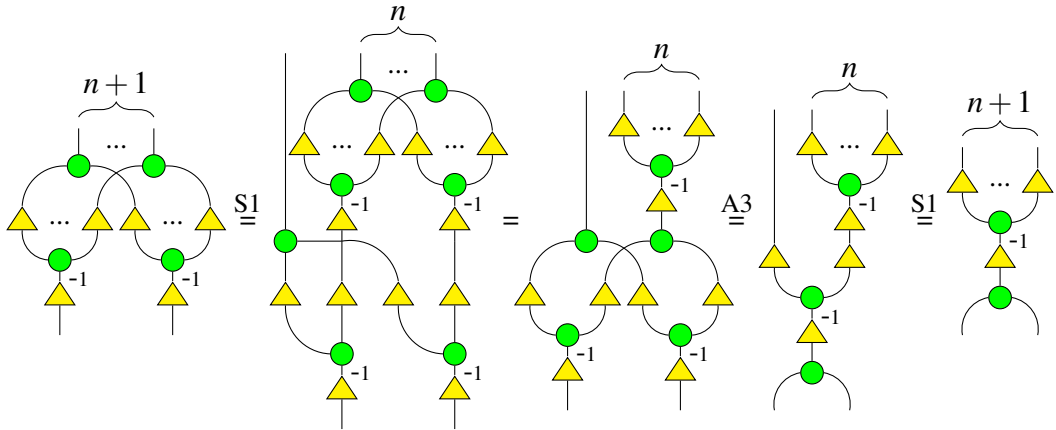
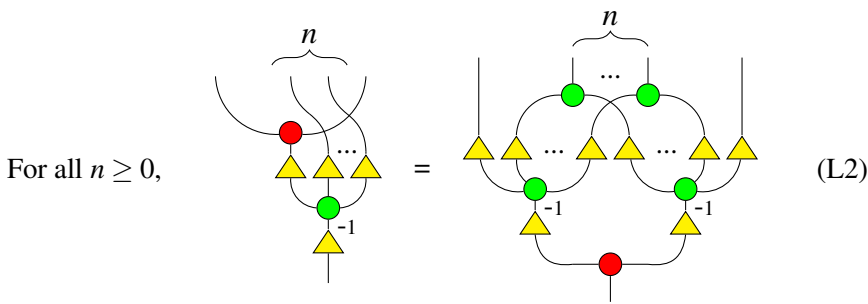
$$
	\bigskip\bigskip

\subsubsection*{Lemma 2}

	For all $n\geq0$, \hspace{0.5cm} \input{TikZit/appendixL2.tikz} \hspace{0.5cm} (L2) \bigskip

	Proof. The $n=0$ case follows from S2 and T3. We obtain all other cases by applying Lemma 1: 
		$$ %
	\beginpgfgraphicnamed{TikZit/appendixL2_Proof1}
	\InputIfFileExists{TikZit/appendixL2_Proof1.tikz}{}{\input{./figures/TikZit/appendixL2_Proof1.tikz}}%
	\endpgfgraphicnamed
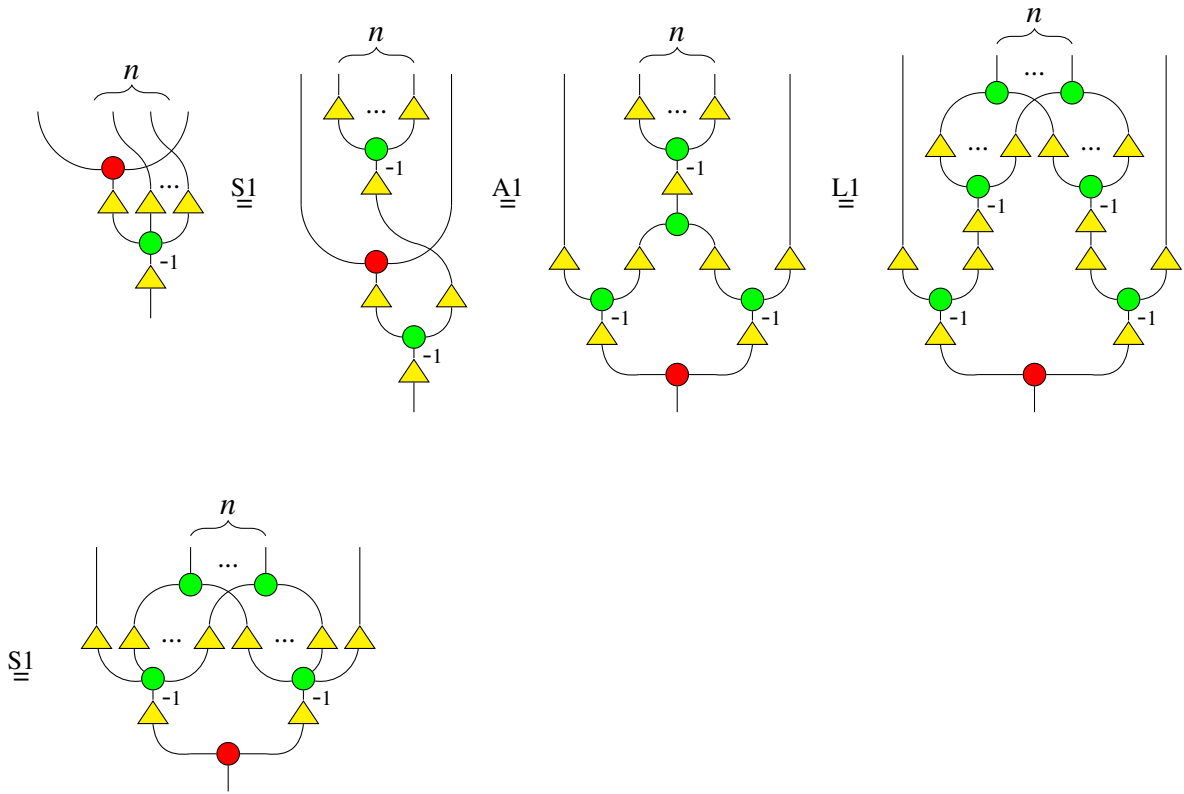
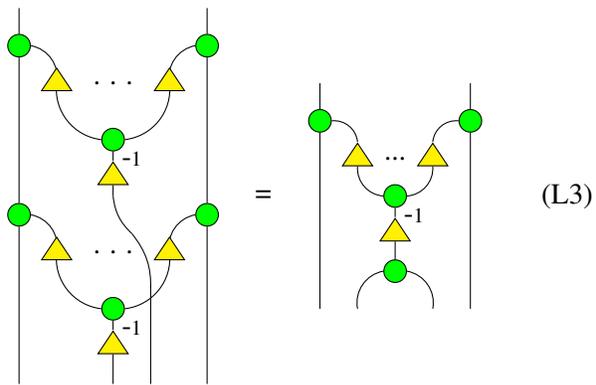
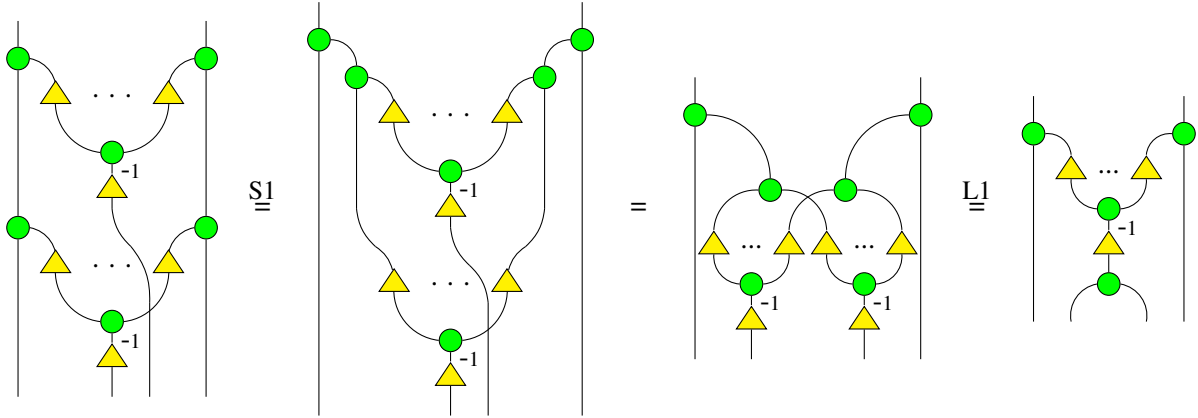
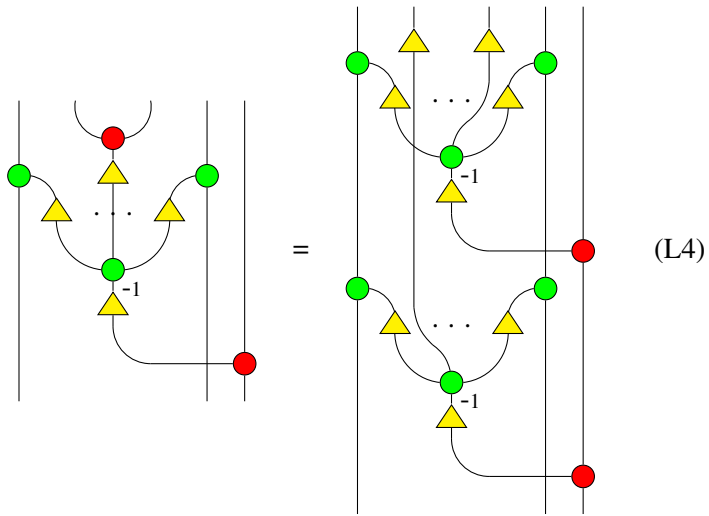
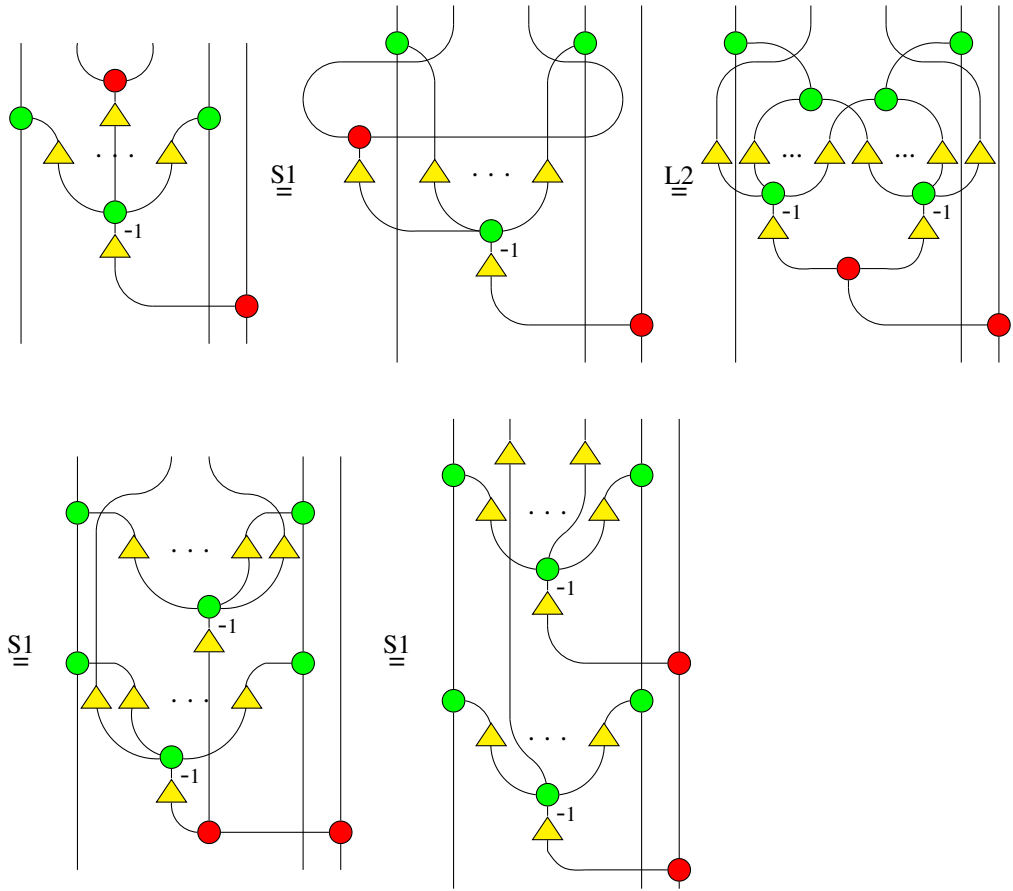
$$
	\input{TikZit/appendixL2_Proof2.tikz} 
	\bigskip\bigskip

\subsubsection*{Lemma 3}

	\input{TikZit/appendixL3.tikz} \hspace{0.5cm} (L3) \bigskip

	Proof. 
		$$ %
	\beginpgfgraphicnamed{TikZit/appendixL3_Proof}
	\InputIfFileExists{TikZit/appendixL3_Proof.tikz}{}{\input{./figures/TikZit/appendixL3_Proof.tikz}}%
	\endpgfgraphicnamed
$$
	 \bigskip\bigskip

\subsubsection*{Lemma 4}

	\input{TikZit/appendixL4.tikz} \hspace{0.5cm} (L4) \bigskip

	\begin{proof}
	\[
	\beginpgfgraphicnamed{TikZit/appendixL4_Proof}
	\InputIfFileExists{TikZit/appendixL4_Proof.tikz}{}{\input{./figures/TikZit/appendixL4_Proof.tikz}}%
	\endpgfgraphicnamed
  
\]
	\end{proof}
\subsubsection*{Proof of Theorem \ref{decomposetheorem}}	
First, we need the following lemmas.
\begin{lemma} For any  $\alpha \in [0, 2\pi)$, we have:
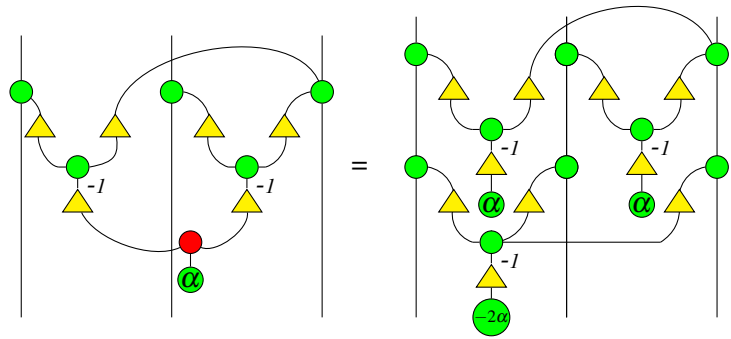
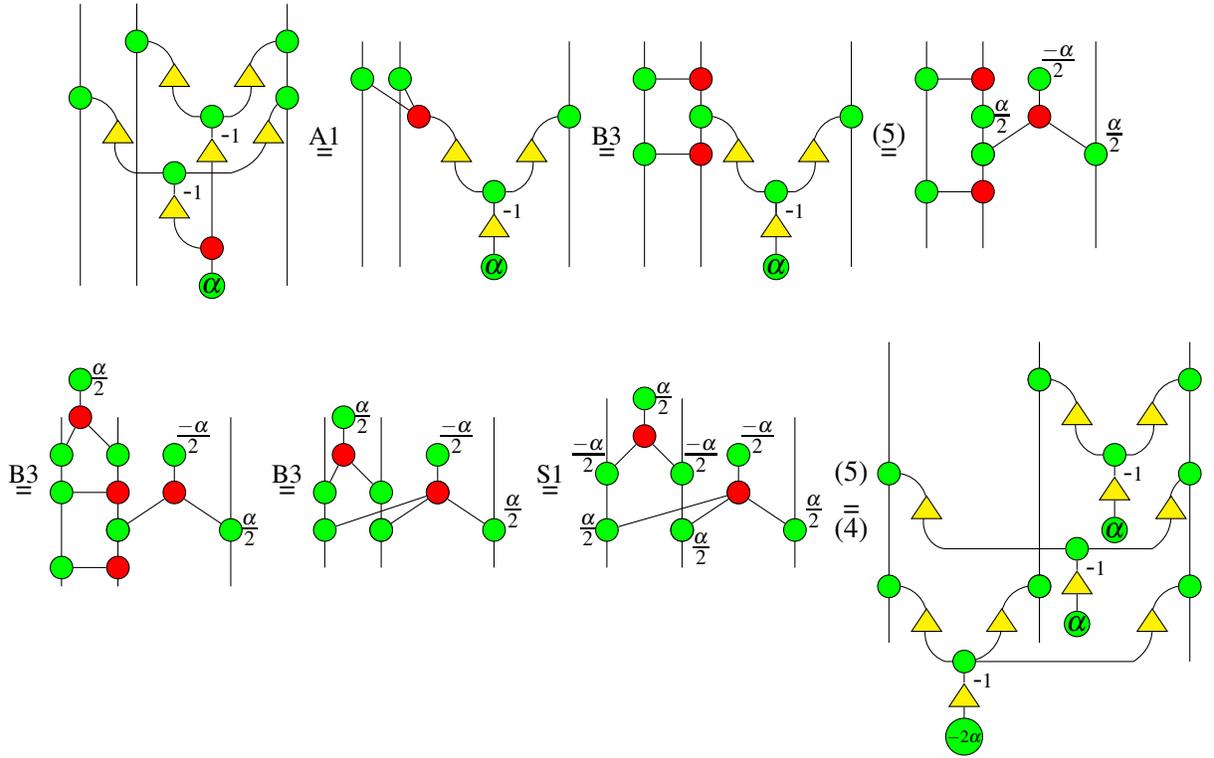
\begin{equation}\label{a4equiv}
	\beginpgfgraphicnamed{TikZit/a4equivalent}
	\InputIfFileExists{TikZit/a4equivalent.tikz}{}{\input{./figures/TikZit/a4equivalent.tikz}}%
	\endpgfgraphicnamed
  
\end{equation}
\end{lemma}

\begin{proof}
\[
	\beginpgfgraphicnamed{TikZit/a4equivalentprf1}
	\InputIfFileExists{TikZit/a4equivalentprf1.tikz}{}{\input{./figures/TikZit/a4equivalentprf1.tikz}}%
	\endpgfgraphicnamed
  
\]
\end{proof}

\begin{lemma} For any  $\alpha \in [0, 2\pi), ~k\geq 1$, we have
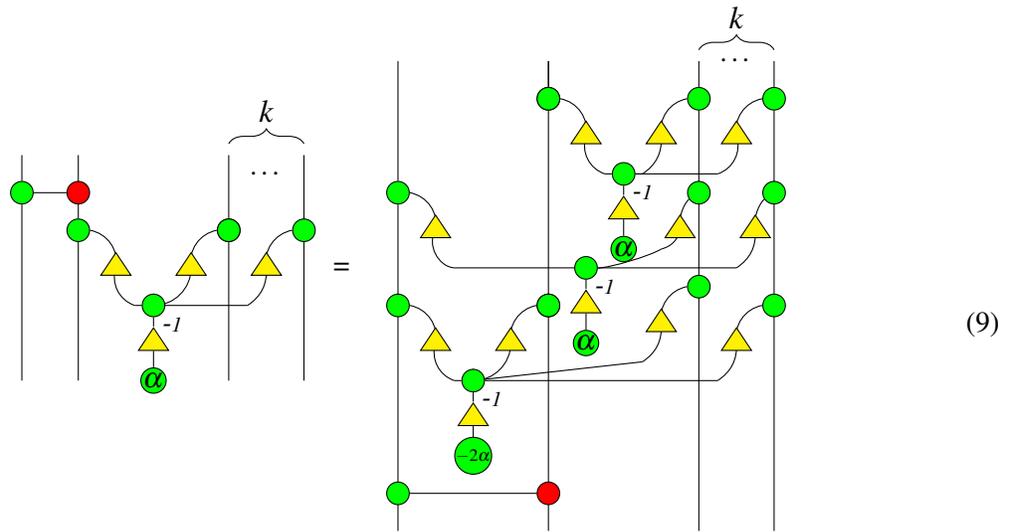
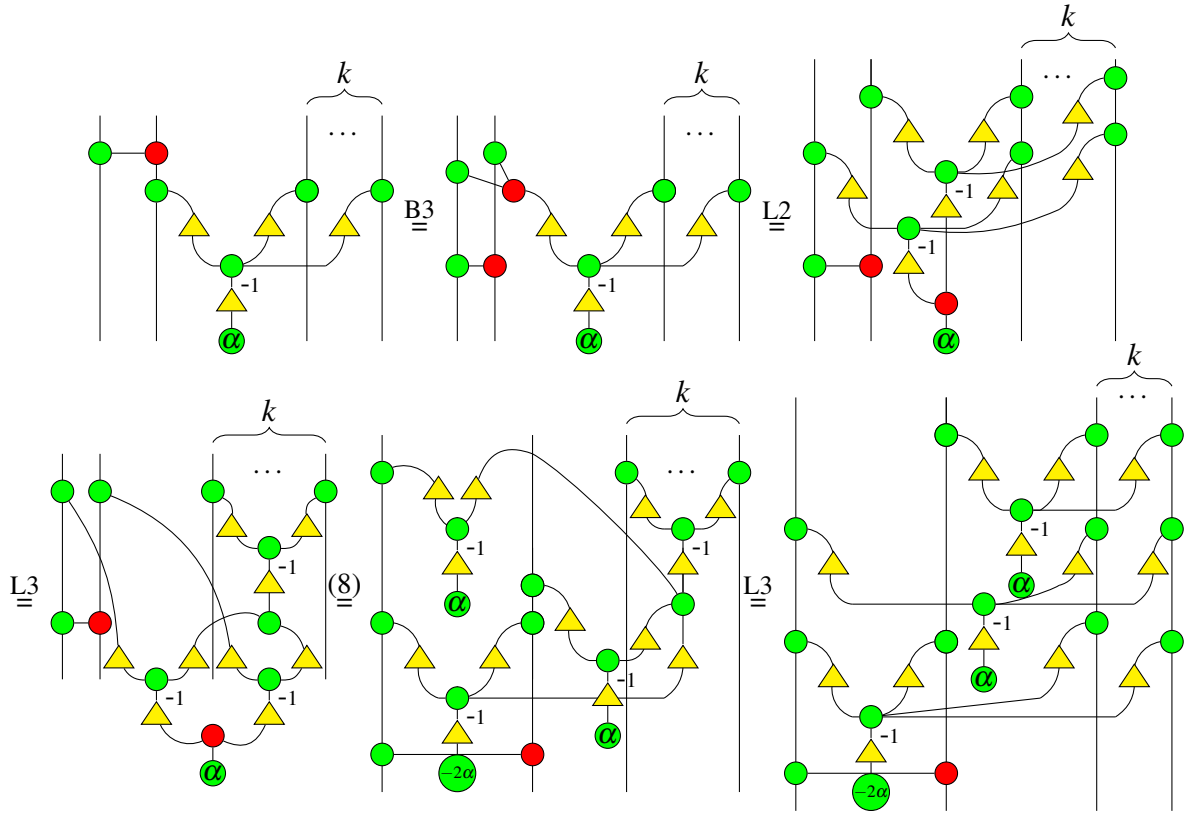
\begin{equation}\label{induction}
	\beginpgfgraphicnamed{TikZit/inductessential}
	\InputIfFileExists{TikZit/inductessential.tikz}{}{\input{./figures/TikZit/inductessential.tikz}}%
	\endpgfgraphicnamed
  
\end{equation}
\end{lemma}

\begin{proof}
\[
	\beginpgfgraphicnamed{TikZit/inductessentialprf2}
	\InputIfFileExists{TikZit/inductessentialprf2.tikz}{}{\input{./figures/TikZit/inductessentialprf2.tikz}}%
	\endpgfgraphicnamed
  
\]
\end{proof}
Now we are ready to prove Theorem \ref{decomposetheorem}.
\begin{proof}
We prove by induction on $n$. If $n=2$, it is just the equality (\ref{gadget-and-2line}) which can be derived from (\ref{gadget-and}). Assume that (\ref{phasedecomfull}) holds for $n=m$. Then for $n=m+1$, we have 
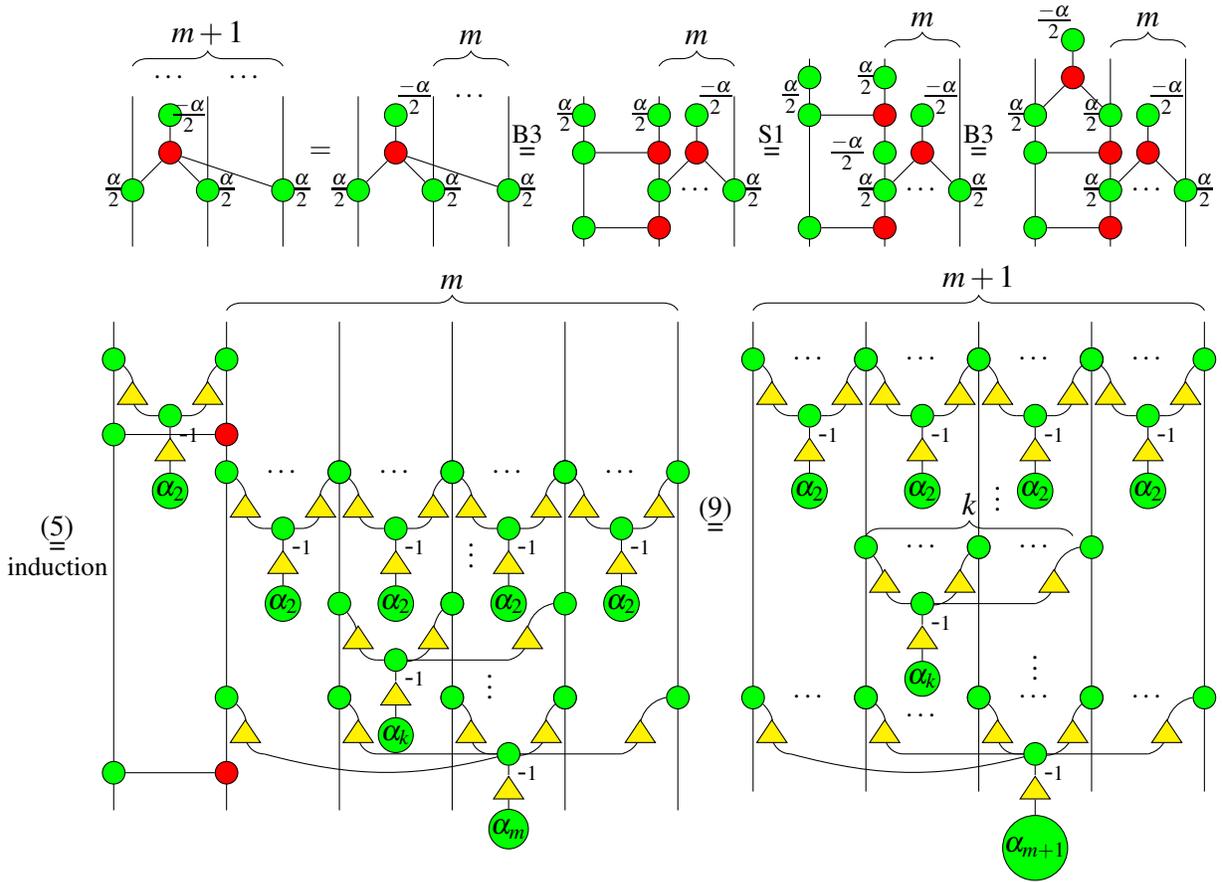
\[
	\beginpgfgraphicnamed{TikZit/phasedecomfullprf1}
	\InputIfFileExists{TikZit/phasedecomfullprf1.tikz}{}{\input{./figures/TikZit/phasedecomfullprf1.tikz}}%
	\endpgfgraphicnamed
  
\]
Therefore, (\ref{phasedecomfull}) holds for $n=m+1$. This completes the proof. 
\end{proof}

\subsubsection*{Proof of Theorem \ref{qbccomplete}}	
\begin{proof}
It suffices to prove that each of the six transformation rules, expressed as ZX-diagrams, are derivable using only the rule set given in the tables above. 

(1) $[t_1, C_1]\cdot[t_1, C_1]\iff\varepsilon$ 

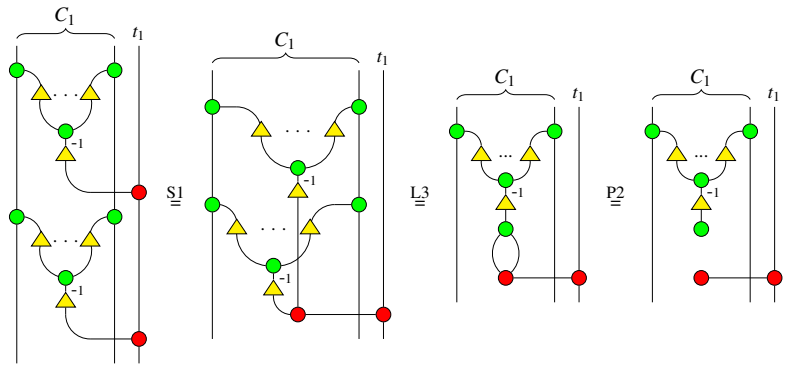
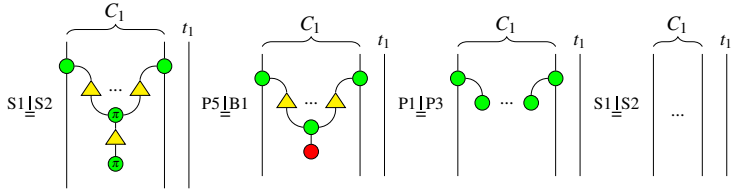
	\scalebox{0.65}{\input{TikZit/QBCRule1-1.tikz}}\\
	\scalebox{0.65}{\input{TikZit/QBCRule1-2.tikz}} \\

(2) $[t_1, C_1]\cdot[t_2, C_2]\iff[t_2, C_2]\cdot[t_1, C_1]$,  if $t_1\not\in C_2$ and $t_2\not\in C_1$

This is a simple consequence of (S1). The case in which $t_1\neq t_2$ is shown below. 

	\scalebox{0.65}{ \input{TikZit/QBCRule2.tikz}}
	 

(3) $[t_1, C_1]\cdot[t_2, C_2]\iff[t_2, C_2]\cdot[t_1, C_1]\cdot[t_1, C_1\cup C_2-\{t_2\}]$ if $t_1\not\in C_2$ and $t_2\in C_1$ \newline

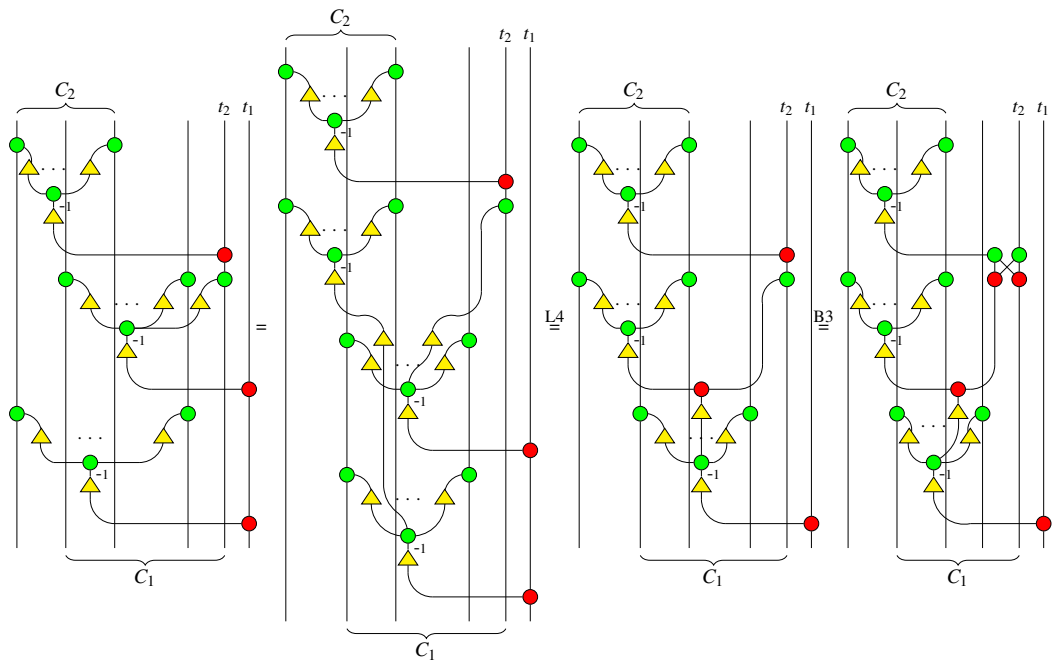
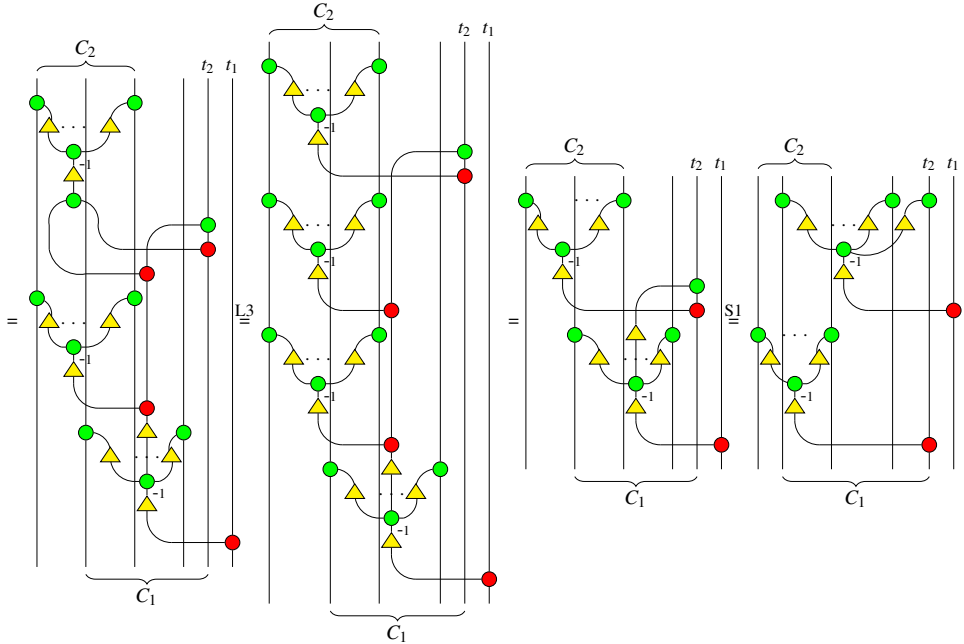
	\scalebox{0.65}{\input{TikZit/QBCRule3-1.tikz}}\\
	\scalebox{0.65}{\input{TikZit/QBCRule3-3.tikz}}\\
	 \bigskip
	 The second-to-last equality comes from rule 1. \bigskip\bigskip

(4) $[t_1, C_1]\cdot[t_2, C_2]\iff[t_2, C_1\cup C_2-\{t_1\}]\cdot[t_2, C_2]\cdot[t_1, C_1]$ if $t_1\in C_2$ and $t_2\not\in C_1$ 

This follows immediately from rule 3: 

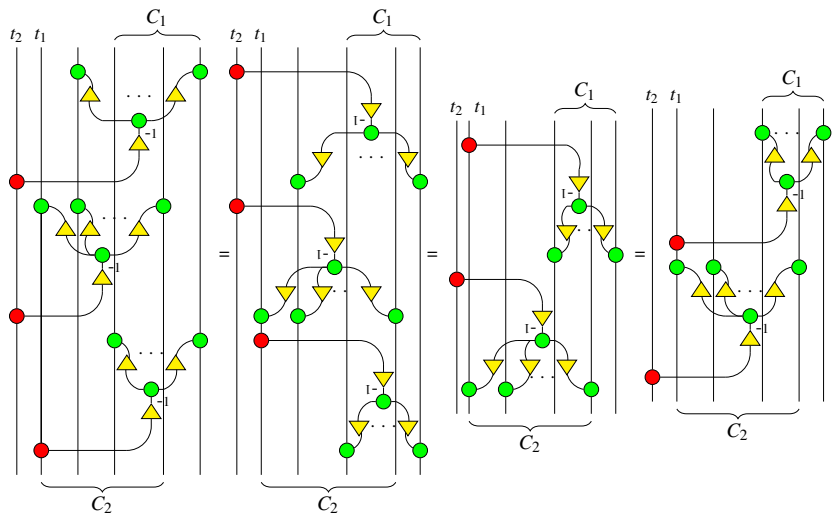
\scalebox{0.65}{\input{TikZit/QBCRule4.tikz}}

(5) $[t_1, \{c_1\}]\cdot[t_2, C_2\cup\{c_1\}]\iff[t_1, \{c_1\}]\cdot[t_2, C_2\cup\{t_1\}]$ if $t_1>n+1$ and no CNOT$_{t_1}$ before $[t_1, \{c_1\}]$ \newline\smallskip

Note that we must assume $t_1\neq t_2$ for the generalized CNOT gate $[t_2, C_2\cup\{t_1\}]$ to be well-defined.

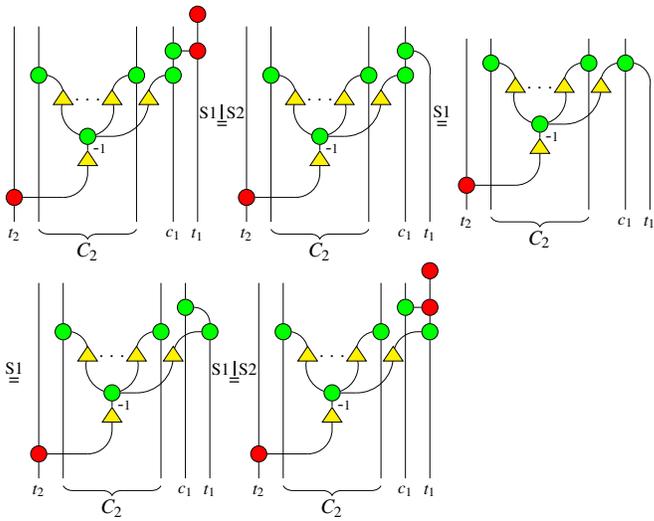
\scalebox{0.65}{\input{TikZit/QBCRule5.tikz}}


(6) $[t, C]\iff\varepsilon$ if there is an integer $i$ such that $i\in C$, $i>n+1$, and there is no CNOT$_i$ before $[t,C]$ 

\scalebox{0.65}{\input{TikZit/QBCRule6.tikz}}

\end{proof}


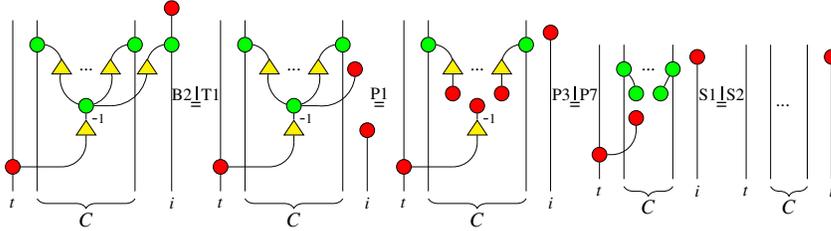
\end{document}